\documentclass[12pt]{article}
\usepackage{mydef}
\usepackage[hide]{ed}
\usepackage{appendix}

\begin{document}

\title{Efficient Solution of Backward Jump-Diffusion PIDEs with Splitting and Matrix Exponentials.
\thanks{Opinions expressed in this paper are those of the author, and do not necessarily reflect the views of Numerix LLC.}}
\author{Andrey Itkin}

\affil{\small Numerix LLC, \\ 150 East 42nd Street, 15th Floor, New York, NY 10017, USA}
\affil{and}
\affil{\small Polytechnic Institute of New York University, \\
6 Metro Tech Center, RH 517E, Brooklyn NY 11201, USA \\
aitkin@poly.edu}

\date{\today}

\maketitle

\begin{abstract}
We propose a new, unified approach to solving jump-diffusion partial integro-differential equations (PIDEs) that often appear in mathematical finance. Our method consists of the following steps. First, a second-order operator splitting on financial processes (diffusion and jumps) is applied to these PIDEs. To solve the diffusion equation, we use standard finite-difference methods, which for multi-dimensional problems could also include splitting on various dimensions. For the jump part, we transform the jump integral into a pseudo-differential operator. Then for various jump models we show how to construct an appropriate first and second order approximation on a grid which supersets the grid that we used for the diffusion part. These approximations make the scheme to be unconditionally stable in time and preserve positivity of the solution which is computed either via a matrix exponential, or via P{\'a}de approximation of the matrix exponent. Various numerical experiments are provided to justify these results.
\end{abstract}

\section{Introduction}
Partial integro-differential equations (PIDEs) naturally appear in mathematical finance if an underlying stochastic process is assumed to be a combination of diffusion and jumps. A wide class of L{\'e}vy processes fall into this category. In modern popular models such as stochastic volatility or, e.g., hybrid models, jumps could accompany any stochastic factor, thus increasing the overall complexity of the problem. For more details about jump-diffusion processes, see \cite{ContTankov, Sato:99}.

Unsurprisingly, most of these PIDEs cannot be solved in closed form. At the same time, a numerical counterpart must be efficient. This is especially important if such a jump-diffusion model is used not only for pricing (given the values of the model parameters), but for their calibration as well. While the solution of the diffusion part (PDE) has been numerously discussed in the literature and various methods were proposed (see, e.g., \cite{fdm2000, Duffy, BrennanSchwartz:1978, ContVolchkova2003, AA2000, HoutFoulon2010}), little can be found for the jump part, which according to the L{\'e}vy-Khinchine formula is represented by a non-local integral.

In this paper we also do not consider jump-diffusion models where the characteristic function (CF) is known in closed form, since then transform methods (FFT, cosine, wavelets etc.) seem to be the most efficient ones. We draw our attention to some particular settings where both the CF and pdf of the diffusion part are not known, while for the jump part the CF can be obtained in closed form. Various popular models are collected under such an umbrella, i.e. local volatility models with jumps, local stochastic volatility models with jumps, etc.

A number of methods were proposed to address the construction of an efficient algorithm for solving these type of PIDEs, see \cite{CarrMayo, Strauss2006, ItkinCarr2012Kinky} and references therein as well as discussion of problems related to their implementation. In particular, they include a discretization of the PIDE that is implicit in the differential terms and explicit in the integral term (\cite{ContVolchkova2003}), Picard iterations for computing the integral equation (\cite{Halluin2004, Halluin2005a}) and a second-order accurate, unconditionally-stable operator splitting (ADI) method that does not require an iterative solution of an algebraic equation at each time step (\cite{AA2000}). Various forms of operator splitting technique were also used for this purpose (\cite{ItkinCarr2012Kinky}). In this paper, we will review operator splitting on financial processes in more detail.

Assuming that an efficient discretization of the PIDE in time was properly chosen, the remaining problem is a fast computation of the jump integral, as it was observed to be relatively expensive. We mention three different approaches to numerical computation of this integral.\footnote{For some models it can be computed analytically, so in what follows we do not take these models into account.}

The first approach assumes a direct approximation of the integral on an appropriate grid and then applies some standard quadrature method, such as Simpson's rule or Gaussian quadrature. This approach may be computationally expensive for two reasons. First, usually the ``jump'' grid is not the same as the ``diffusion'' grid. Therefore, after the integral is computed, its values at the jump grid should be re-interpolated to the diffusion grid. Second, the integral is defined on an infinite domain, so either the domain has to be truncated or a non-uniform grid has to be used. Moreover, the complexity becomes greater if an implicit discretization of the integral is used, because it requires the  solution of a dense system of linear equations of a large size. Therefore, most often an explicit discretization is utilized, which brings  some constraints on the time steps to guarantee stability of the scheme.

However, an exponential change of variables reduces the expense of evaluating the integral at all points. This change converts the integral term into a correlation integral, which can be evaluated at all the grid points simultaneously using a Fast Fourier Transform (FFT). This approach has been suggested by many authors (\cite{fdm2000, AA2000, Wilmott1998}). Still, this could be expensive because a large number of FFT nodes may be required for better accuracy. Another issue is that using FFT to compute a product of matrix $A$ and vector requires $A$ to be circulant, while the matrix obtained after discretization of the jump integral is not of that type. Therefore, a direct (naive) usage of FFT for this purpose produces undesirable so-called ``wrap-around'' errors. A common technique to eliminate these errors is to embed $A$, which is actually a Toeplitz matrix, into a circulant matrix. This, in turn, requires doubling the initial vector of unknowns, which makes the algorithm slower. This approach was improved in \cite{Halluin2005b}, still some extension of the computational region is required in both upper and lower directions while not doubling the grid size. Also linear interpolation with a pre-computed coefficients was proposed to transform option values from the non-uniform diffusion grid to the uniform jump grid, which keeps the second order of approximation, and is efficient performance-wise.

The second approach to computing the jump integral utilizes an alternative representation of this integral in the form of a pseudo-differential operator, which puts  the entire PIDE in the form of a fractional PDE. This problem was considered in \cite{Cartea2007} and \cite{ItkinCarr2012Kinky}. A recent survey of the existing literature on this subject and techniques for computation of the jump integral using the Grunwald-Letnikov approximation (which is of the first order in space) is given in \cite{AndersenLipton2012}. As it is known from \cite{AbuSaman2007, MeerschaertTadjeran2004, Tadjeran2006, MeerschaertTadjeran2006, Sousa2008}, a standard Grunwald-Letnikov approximation leads to unconditionally unstable schemes. To improve this, a shifted Grunwald-Letnikov approximation was proposed, which allows construction of an unconditionally stable scheme of the first order in space.\footnote{A second-order approximation could in principle be constructed as well, but this would result in a massive calculation for the coefficients. Therefore, this approach was not further elaborated on.} However, solving pricing equations to second order  in the space variable is almost an industry standard, and therefore this method requires further investigation to address this demand.

The third method exploits a nice idea first proposed in \cite{CarrMayo}.  Carr and Mayo found that for some L{\'e}vy models, the solution of the integral evolutionary equation\footnote{This equation naturally arises at some step of the splitting procedure, if splitting is organized by separating diffusion from jumps.} is equivalent to the solution of a particular PDE. The problem is then to find a proper space-differential operator (kernel) to construct such a PDE. Carr and Mayo demonstrated the advantage of this approach for the Merton and Kou models, and showed which parabolic equations provide the necessary solution. Later in \cite{ItkinCarr2012Kinky}, this idea was further generalized to the class of pseudo-parabolic equations as applied to a class of L\'evy processes, known as GTSP/CGMY/KoBoL models. These pseudo-parabolic equations could be formally analytically solved via a matrix exponential. Itkin and Carr then discuss a numerical method to efficiently compute this matrix exponential. When the parameter $\alpha$ of the GTSP/CGMY/KoBoL model is an integer, this method uses a finite-difference scheme similar to those used for solving parabolic PDEs, and the matrix of this finite-difference scheme is banded. Therefore, in this case, the computation of the jump integral:
\begin{itemize}
\item Is provided on the same grid as was constructed for the diffusion (parabolic) PDE. Outside of this domain (if ever needed, e.g. for European options),  the PIDE grid is further extended to an infinite domain,\footnote{In other words the PIDE grid is a superset of the corresponding PDE grid.} but no interpolation is required afterwards.
\item At every time step it has linear ($O(N)$) complexity in the number of the grid nodes $N$, since the results (e.g., option prices) are given by solving a linear system of equations with a banded matrix. In the case of a real parameter $\alpha$, Itkin and Carr suggested computing the prices using the above algorithm at three values of an integer $\widetilde{\alpha}$ closest to the given real $\alpha$, and then interpolating using any interpolation of the second order.
\end{itemize}

In this paper we use a different flavor of this idea. First, we use an operator-splitting method on the financial processes, thus separating the computation of the diffusion part from the integral part. Then, similar to \cite{ItkinCarr2012Kinky}, we represent the jump integral in the form of a pseudo-differential operator. Next we formally solve the obtained evolutionary partial pseudo-differential equations via a matrix exponential. We then show that the matrix exponential can be efficiently computed for many popular L{\'e}vy models, and that the efficiency of this method could be not worse than that of the FFT. The proposed method is almost universal, i.e., allows computation of PIDEs for various jump-diffusion models in a unified form. We also have to mention that this method is relatively simple for implementation.

Note, that the idea of this method is in some sense close to another popular approach - Fourier Space Time-Stepping Method (FSTS), see \cite{Surkov2007}. The advantage of both approaches lies in the fact that the jump integral could contain singularities, while in FSTS and in the present approach these singularities are integrated out. The difference is that in FSTS this is done by switching to the Fourier space (one FFT). Then FSTS uses finite-difference method directly in the Fourier space, and finally switches back to the price space (one inverse FFT) at each time step. In our approach we eliminate these two extra FFTs at every time step since we are working in the price space all the time. The second difference is that FSTS treats both diffusion and jump operators in a symmetric way by switching all calculations to the Fourier space. This, however, can not be done, if the CF of the diffusion operator is not known, for instance for LV or LSV models, while this is not a limitation for our approach.

Let us also mention one more method proposed by \cite{LiptonSepp2009a} as applied to the Stern-Stern model. Though it is not evident how to generalize this method for other models, it provides a very efficient computational algorithm for this particular model. Also for simpler jump models, like that of Merton and Kou, there are some other efficient methods  in  the literature, see, e.g., \cite{Tangman2008a, Lee2012}.

Based on the above survey, we can conclude that in our domain of models (LV + jumps, LSV + jumps, etc.) the most relevant predecessors of our work are \cite{Halluin2004} in general, and for Merton and Kou models - \cite{CarrMayo}. Therefore, we want to underline the differences between these approaches and that in this paper:
\begin{enumerate}
\item For Merton's model following our general approach we re-derive the result of \cite{CarrMayo}.
\item For Kou's model again following our general approach we derive a different flavor of \cite{CarrMayo}.
\item For CGMY model to get a second order approximation in time \cite{Halluin2004} use Picard iterations, while here we provide two flavors of the method: one is similar to
    \cite{Halluin2004} and uses FFT to multiply matrix by vector (with the complexity $O(N \log N)$; the other one, which doesn't need iterations, exploits matrix exponential (with the complexity $O(N^2)$). In the latter case our method also doesn't need to extend an FFT grid to avoid wrap-around effects.
\item For CGMY model method of \cite{Halluin2004} experiences some difficulties when parameter $\alpha$ of the CGMY model is close to 2, see \cite{WangWanForsyth2007}. Here we show what is the source of this problem and propose another method to address this issue.
\item For CGMY model a special treatment of the area close to $x=0$ is required, see \cite{ContVolchkova2003, WangWanForsyth2007}. Here there is no such a problem due to an analytical representation of the jump integral in the form of a pseudo-differential operator (i.e, this singularity is integrated out analytically).
\end{enumerate}

As far as the complexity of the proposed methods is concerned let us mention the following.
    \begin{itemize}
    \item For Merton's and Kou's jumps the algorithms that reduce the total complexity to $O(N)$ per time step are presented in this paper. For the Merton's jumps this includes a new idea to use Fast Gauss Transform, \cite{IFGT} instead of the finite difference method when solving the intermediate heat equation.
    \item For CGMY model with $\alpha < 0$ it is clear that the method is very similar to that of \cite{WangWanForsyth2007} since our experiments show that the matrix exponential is less efficient than FFT in this case, and the L\'evy kernel doesn't have singularities. Instead we
        recommended to use a different flavor of this method, see \cite{ItkinCarr2012Kinky} is more efficient with complexity $O(N)$.
    \item For CGMY model with $0 < \alpha < 1$ again the method is similar to that of \cite{WangWanForsyth2007}. However, in our case it doesn't require a special treatment  of the point $x=0$ because this singularity was already integrated out. On the other hand, in this region we provide only a first order scheme $O(h)$ leaving extension of the method to $O(h^2)$ as an open yet problem. From this prospective in this region the method of \cite{WangWanForsyth2007} is more accurate. Again, approach of \cite{ItkinCarr2012Kinky} would improve it if one uses the approach of this paper to compute the price at some $1 < \alpha < 2$, and then use it in the interpolation procedure of \cite{ItkinCarr2012Kinky}. The total complexity for the entire algorithm then coincide with the complexity in the case $1 < \alpha < 2$.
    \item For CGMY model with $1 < \alpha < 2$ our method has some advantage as compared with that of \cite{WangWanForsyth2007}, namely: i) computation of the matrix exponential eliminates the necessity for Picard iterations which poorly converge in this case (we also explain why a slow convergence is observed in the latter approach), and ii) the singularity close to $\alpha = 2$ is already integrated out, and, therefore, the method works fine in this case even at $\alpha $ close to 2 (the results in the last section are provided for $\alpha = 1.95$).
    \end{itemize}

The rest of the paper is organized as follows. In section~\ref{Sec2} we briefly discuss a general form of a backward PIDE for the class of L{\'e}vy models. In Section~\ref{Sec3}, we introduce a splitting technique for nonlinear operators. In section~\ref{Sec4}, we present our general approach to the solution of the PIDE using a splitting and matrix exponential approach. An explicit construction of various finite-difference schemes of the first and second order is presented in the next section. There we consider the following jump models: Merton, Kou and GTSP (also known as CGMY or KoBoL). The results presented in the last two sections are new, and to the best of our knowledge have not been discussed in the literature. Our technique utilizes some results from matrix analysis related to definitions of M-matrices, Metzler matrices and eventually exponentially nonnegative matrices. We also give the results of various numerical tests to demonstrate convergence of our method. In section~\ref{Sec6}, some additional numerical examples are presented that consider all steps of the splitting algorithm, not just the jump part as in the previous sections. The final section concludes.

\section{L{\'e}vy Models and Backward PIDE } \label{Sec2}
To avoid uncertainty, let us consider the problem of pricing equity options written on a single stock. As we will see, this specification does not cause us to lose any generality, but it makes the description more practical. We assume an underlying asset (stock) price $S_t$ be driven by an exponential of a L{\'e}vy process
\begin{equation} \label{Levy}
S_t = S_0 \exp (L_t), \quad 0 \le t \le T,
\end{equation}
where $t$ is time, $T$ is option expiration, $S_0 = S_t \ |_{t=0}$, $L_t$ is the L{\'e}vy process $L = (L_t)_{0 \le t \le T}$ with a nonzero Brownian (diffusion) part. Under the pricing measure, $L_t$ is given by
\begin{equation} \label{Lt}
L_t = \gamma t  + \sigma W_t + Y_t, \qquad \gamma, \sigma \in \mathbb{R}, \quad \sigma > 0,
\end{equation}
with  \LY triplet $(\gamma, \sigma, \nu$), where $W_t$ is a standard Brownian motion on $0 \le t \le T$ and $Y_t$ is a pure jump process.

We consider this process under the pricing measure, and therefore $e^{-(r-q) t} S_t$  is a martingale, where $r$ is the interest rate and $q$ is a continuous dividend. This allows us to express $\gamma$ as (\cite{Eberlein2009})
\[
\gamma = r - q - \frac{\sigma^2}{2} - \int_\mathbb{R} \left(e^x -1 -x {\bf 1}_{|x| < 1}\right)\nu(dx),
\]
\noindent where $\nu(dx)$ is a \LY measure which satisfies
\[ \int_{|x| > 1}e^x \nu(dx) < \infty.  \]

We leave $\nu(dx)$ unspecified at this time, because we are open to consider all types of jumps including those with finite and infinite variation, and finite and infinite activity.
\footnote{We recall that a standard Brownian motion already has paths of infinite variation. Therefore, the \LY process in \eqref{Lt} has infinite variation since it contains a continuous martingale component. However, here we refer to the infinite variation that comes from the jumps.}

To price options written on the underlying process $S_t$, we want to derive a PIDE that describes time evolution of the European option prices $C(x,t), \ x \equiv \log (S_t/S_0)$. Using a standard martingale approach, or by creating a self-financing portfolio, one can derive the corresponding PIDE (\cite{ContTankov})
\begin{multline} \label{PIDE}
r C(x,t) = \fp{C(x,t)}{t} + \left(r-\frac{1}{2}\sigma^2 \right) \fp{C(x,t)}{x} + \frac{1}{2}\sigma^2 \sop{C(x,t)}{x} \\
+  \int_\mathbb{R}\left[ C(x+y,t) - C(x,t) - (e^y-1)\fp{C(x,t)}{x} \right] \nu(dy)
\end{multline}
for all $(x,t) \in \mathbb{R} \times (0,T)$, subject to the terminal condition
\begin{equation}
C(x,T) = h(x),
\end{equation}
where $h(x)$ is the option payoff, and some boundary conditions which depend on the type of the option. The solutions of this PIDE usually belong to the class of viscosity solutions (\cite{ContTankov}).

We now rewrite the integral term using the following idea. It is well known from quantum mechanics (\cite{OMQM}) that a translation (shift) operator in $L_2$ space could be represented as
\begin{equation} \label{transform}
    \mathcal{T}_b = \exp \left( b \dfrac{\partial}{\partial x} \right),
\end{equation}
with $b$ = const, so
\[ \mathcal{T}_b f(x) = f(x+b). \]

Therefore, the integral in Eq.~(\ref{PIDE}) can be formally rewritten as
\begin{align} \label{intGen}
\int_\mathbb{R} \left[ C(x+y,t) \right. & \left. - C(x,t) - (e^y-1) \fp{C(x,t)}{x} \right] \nu(dy) =  \mathcal{J} C(x,t), \\
\mathcal{J} & \equiv \int_\mathbb{R}\left[
\exp \left( y \dfrac{\partial}{\partial x} \right) - 1 - (e^y-1) \fp{}{x} \right] \nu(dy). \nonumber
\end{align}

In the definition of operator $\mathcal{J}$ (which is actually an infinitesimal generator of the jump process), the integral can be formally computed under some mild assumptions about existence and convergence if one treats the term $\partial/ \partial x$ as a constant. Therefore, operator $\mathcal{J}$ can be considered as some generalized function of the differential operator $\partial_x$. We can also treat $\mathcal{J}$ as a pseudo-differential operator.

With allowance for this representation, the whole PIDE in the \eqref{PIDE} can be rewritten in operator form as
\begin{equation} \label{oper}
\partial_\tau C(x,\tau) = [\mathcal{D} + \mathcal{J}]C(x,\tau),
\end{equation}
\noindent where $\tau = T - t$ and $\mathcal{D}$ represents a differential (parabolic) operator
\begin{equation}
\mathcal{D} \equiv - r + \left(r-\frac{1}{2}\sigma^2 \right) \fp{}{x} + \frac{1}{2}\sigma^2 \sop{}{x},
\end{equation}
 where the operator $\mathcal{D}$ is an infinitesimal generator of diffusion.

Notice that for jumps with finite variation and finite activity, the last two terms in the definition of the jump integral $\mathcal{J}$ in \eqref{PIDE} could be integrated out and added to the definition of $\mathcal{D}$. In the case of jumps with finite variation and infinite activity, the last term could be integrated out. However, here we will leave these terms under the integral for two reasons: i) this transformation (moving some terms under the integral to the diffusion operator) does not affect our method of computation of the integral, and ii) adding these terms to the operator $\mathcal{D}$ negatively influences the stability of the finite-difference scheme used to solve the parabolic equation
$\mathcal{D} C(x,t) = 0$. This equation naturally appears as a part of our splitting method, which is discussed in the next section.

\section{Operator Splitting Technique}\label{Sec3}
To solve Eq.~(\ref{oper}) we use splitting. This technique is also known as the method of fractional steps (see \cite{yanenko1971, samarskii1964, dyakonov1964}) and sometimes is cited in financial literature as Russian splitting or locally one-dimensionally schemes (LOD) (\cite{Duffy}).

The method of fractional steps reduces the solution of the original $k$-dimensional unsteady problem to the solution of $k$ one-dimensional equations per time step. For example, consider a two-dimensional diffusion equation with a solution obtained by using some finite-difference method. At every time step, a standard discretization on space variables is applied, such that the finite-difference grid contains $N_1$ nodes in the first dimension and $N_2$ nodes in the second dimension. Then the problem is solving a system of $N_1 \times N_2$ linear equations, and the matrix of this system is block-diagonal. In contrast, utilization of splitting results in, e.g.,  $N_1$ systems of $N_2$ linear equations, where the matrix of each system is banded (tridiagonal). The latter approach is easy to implement and, more importantly, provides significantly better performance.

The previous procedure uses operator splitting in different dimensions. \cite{marchuk1975} and then \cite{Strang} extended this idea for complex physical processes (for instance, diffusion in the chemically reacting gas, or the advection-diffusion problem). In addition to (or instead of) splitting on spatial coordinates, they also proposed splitting the equation into physical processes that differ in nature, for instance, convection and diffusion. This idea becomes especially efficient if the characteristic times of evolution (relaxation time) of such processes are significantly different.

For a general approach to splitting techniques for {\it linear} operators using Lie algebras, we refer the reader to \cite{LanserVerwer}. Consider an equation
\begin{equation} \label{operEq}
\fp{f({\bf x},t)}{t} = \mathcal{L} f({\bf x},t)
\end{equation}
\noindent where $f({\bf x},t)$ is some function of independent variables ${\bf x},t, \ {\bf x} = x_1 ... x_k$, and $\mathcal{L}$ is some linear k-dimensional operator in ${\bf x}$ space.

Decomposing the total (compound)  operator $\mathcal{L}$ for problems of interest seems natural if, say, $\mathcal{L}$ can be represented as a sum of $k$ noncommuting linear operators $\sum_{i=1}^k \mathcal{L}_i$. In this case the operator equation \eqref{operEq}can be formally integrated via an operator exponential, i.e., \[ f({\bf x},t) = e^{t \mathcal{L}} f({\bf x}, 0)  = e^{ t \sum_{i=1}^k \mathcal{L}_i} f(0). \]
Due to the noncommuting property, the latter expression can be factorized into a product of operators
\[ f({\bf x},t) = e^{t \mathcal{L}_k} ... e^{t \mathcal{L }_1}f({\bf x},0).\]
This equation can then be solved in $N$ steps sequentially by the following procedure:
\begin{align*}
f^{(1)}({\bf x},t) &= e^{t \mathcal{L}_1}f({\bf x}, 0), \nonumber \\
f^{(2)}({\bf x},t) &= e^{t \mathcal{L}_2}f^{(1)}({\bf x},t), \nonumber \\
&\quad \vdots \\
f^{(k)}({\bf x},t) &= e^{t \mathcal{L}_k}f^{(k-1)}({\bf x},t), \nonumber \\
f({\bf x}, t) &= f^{(k)}({\bf x},t). \nonumber
\end{align*}

This algorithm is exact (no bias) if all the operators commute. If, however, they do not commute, the above algorithm provides only a first-order approximation in time (i.e., $O(t)$) to the exact solution.

To get the second-order splitting for noncommuting operators, Strang proposed a new scheme, which in the simplest case ($k=2$) is (\cite{Strang})
\begin{equation}
f({\bf x},t) = e^{t L} f({\bf x},0) = e^{ t (L_1 + L_2)} f({\bf x},0) = e^{ \frac{t}{2} L_1 } e^{t L_2} e^{ \frac{t}{2} L_1 } f({\bf x},0) + O(t^2).
\end{equation}

For parabolic equations with constant coefficients, this composite algorithm is second-order accurate in $t$ provided the numerical procedure that solves a corresponding equation at each splitting step is at least second-order accurate.

The above analysis, however, cannot be directly applied to our problem, because after transformation \eqref{transform} is applied, the jump integral transforms to a non-linear operator \eqref{intGen}.
For {\it non-linear} operators, the situation is more delicate. As shown in \cite{ThalhammerKoch2010}, the theoretical analysis of the nonlinear initial value problem
\[ u'(t) = F(u(t)), \qquad 0 \le t \le T \]
for a Banach-space-valued function $u: [0,T] \rightarrow X$ given an initial condition $u(0)$ could be done using calculus of Lie derivatives. A formal linear representation of the exact solution is
\[ u(t) = \mathcal{E}_F(t,u(0)) = e^{t D_F} u(0), \qquad 0 \le t \le T, \]
where the evolution operator and Lie derivatives are given by
\begin{align*}
e^{t D_F} v &= \mathcal{E}_F(t,v), \quad e^{t D_F} G v = G(\mathcal{E}_F(t,v)), \quad 0 \le t \le T, \\
D_F v &= F(v), \quad D_F G v = G'(v) F(v)
\end{align*}
for an unbounded nonlinear operator $G: D(G) \subset X \rightarrow X$. Using this formalism, \cite{ThalhammerKoch2010} showed that Strang's  second-order splitting method remains unchanged in the case of nonlinear operators.

Using this result for \eqref{oper} gives rise to the following numerical scheme:
\begin{align} \label{splitFin}
C^{(1)}(x,\tau) &= e^{\frac{\Delta \tau}{2} \mathcal{D} } C(x,\tau), \\
C^{(2)}(x,\tau) &= e^{\Delta \tau \mathcal{J}} C^{(1)}(x,\tau) \nonumber, \\
C(x,\tau+\Delta \tau) &= e^{\frac{\Delta \tau}{2} \mathcal{D} } C^{(2)}(x,\tau).  \nonumber
\end{align}

Thus, instead of an unsteady PIDE, we obtain one PIDE with no drift and diffusion (the second equation in \eqref{splitFin}) and two unsteady PDEs (the first and third ones in \eqref{splitFin}).

In what follows, we consider how to efficiently solve the second equation, while assuming that the solution of the first and the third equations can be obtained using any finite-difference method that is sufficiently efficient. To this end, in various examples given in the next sections we will explicitly mention what particular method was used for this purpose.

In this paper, we do not discuss the uniqueness and existence of the solution for the PIDE; to do so would move us to the definition of a viscosity solution for this class of integro-differential equations. For more details, see \cite{ContTankov} and \cite{Arisawa2005}.

Lastly, let us mention that $\mathcal{J} = \phi(-i \partial_x)$, where $\phi(u)$ is the characteristic exponent of the jump process. This directly follows from the L{\'e}vy-Khinchine theorem.

\section{Solution of a Pure Jump Equation} \label{Sec4}
We begin with the following observation. By definition of the jump generator $\mathcal{J}$, under some mild constraints on its existence, $\mathcal{J}$ could be viewed as a function of the operator $\partial_x$. Therefore, solving the integral (second) equation in \eqref{splitFin} requires a few steps.

First, an appropriate discrete grid ${\bf G}(x)$ has to be constructed in the truncated (originally infinite) space domain. This grid could be nonuniform. An important point is that in the space domain where the parabolic equations of \eqref{splitFin} are defined, this grid should coincide with the finite-difference grid constructed for the solution of these parabolic equations.\footnote{So the PIDE grid is a superset of the PDE grid.} This is to avoid interpolation of the solution that is obtained on the jump grid (the second step of the splitting algorithm) to the diffusion grid that is constructed to obtain solutions at the first and third splitting steps.

To make this transparent, let the parabolic equation be solved at the space domain $[x_0,x_k],$ $\ x_0 > -\infty$, $x_k < \infty$ using a nonuniform grid with $k+1$ nodes ($x_0, x_1,...,x_k$) and space steps $h_1 = x_1-x_0, ..., h_k = x_k - x_{k-1}$. The particular choice of $x_0$ and $x_k$ is determined by the problem under consideration. We certainly want $|x_0|$ and  $|x_k|$ not to be too large.  The integration limits of $\mathcal{J}$ in \eqref{intGen} are, however,  plus and minus infinity. Truncation of these limits usually is done to fit memory and performance requirements. On the other hand, we want a fine grid close to the option strike $K$ for better accuracy. Therefore, a reasonable way to construct a jump grid is as follows. For $x_0 \le x \le x_k$, the jump grid coincides with the grid used for solution of the parabolic PDEs. Outside of this domain, the grid is expanded by adding nonuniform steps; i.e., the entire jump grid is $x_{-K}, x_{1-K}, ... x_{-1}, x_0, x_1, ..., x_k, x_{k+1}, ..., x_{k+M}$. Here $K >0, \ M>0$ are some integer numbers that are chosen based on our preferences. Since contribution to $\mathcal{J}$  from very large values of $x$ is negligible, the outer grid points $x_{-K}, x_{1-K}, ... x_{-1}$ and $x_{k+1}, ..., x_{k+M}$ can be made highly nonuniform. One possible algorithm could be to have the steps of these grids be a geometric progression. This allows one to cover the truncated infinite interval with a reasonably small number of nodes.

Second, the discretization of $\partial_x$ should be chosen on ${\bf G}(x)$. We want this discretization to:
\begin{enumerate}
\item Provide the necessary order of approximation of the whole operator $\mathcal{J}$ in space.
\item Provide unconditional stability of the solution of the second equation in \eqref{splitFin}.
\item Provide positivity of the solution.
\end{enumerate}
Let  $\Delta_x$ denote a discrete analog of $\partial_x$ obtained by discretization of $\partial_x$ on the grid ${\bf G}(x)$. Accordingly, let us define the matrix $J(\Delta_x)$ to be the discrete analog of the operator $\mathcal{J}$ on the grid ${\bf G}(x)$.
The following proposition translates the above requirements to the conditions on $J(\Delta_x)$.
\begin{proposition} \label{prop0}
The finite-difference scheme
\begin{equation} \label{fd0}
C(x, \tau + \Delta \tau) = e^{\Delta \tau J(\Delta_x)} C(x,\tau)
\end{equation}
is unconditionally stable in time $\tau$ and preserves positivity of the vector $C(x,\tau)$ if there exists an M-matrix $B$ such that $J(\Delta_x) = - B$.
\end{proposition}
\begin{proof}
By definition of an M-matrix (see \cite{BermanPlemmons1994}), the class of M-matrices contains those matrices whose off-diagonal entries are less than or equal to zero, while all diagonal elements are positive. All eigenvalues of an M-matrix have a positive real part. Therefore, if $B$ is an M-matrix, all eigenvalues of $J(\Delta_x)$  have a negative real part. Therefore, $\|e^{\Delta \tau J(\Delta_x)}\| < 1$ (in the spectral norm), and thus the scheme \eqref{fd0} is unconditionally stable.

Now since $B$ is an M-matrix, $J$ is a Metzler matrix (\cite{BermanPlemmons1994}). An exponential function of the Metzler matrix is a positive matrix. Therefore, if  $C(x,\tau)$ is positive, the scheme \eqref{fd0} preserves the  positivity of $C(x,\tau + \Delta \tau)$. $\blacksquare$
\end{proof}

This proposition gives us a recipe for the construction of the appropriate discretization of the operator $\mathcal{J}$. In the next section, we will give some explicit examples of this approach.

Once the discretization is performed, all we need is  to compute a matrix exponential $ e^{\Delta \tau J(\Delta_x)}$, and then a product of this exponential with $C(x,\tau)$. The following facts make this method competitive with those briefly described in the introduction. We need to take into account that:
\begin{enumerate}
\item The matrix $J(\Delta_x)$ can be precomputed once the finite-difference grid ${\bf G}(x)$ has been built.

\item If a constant time step is used for computations, the  matrix $\mathcal{A} = e^{\Delta \tau J(\Delta_x)}$ can also be precomputed.
\end{enumerate}

If the above two statements are true, the second splitting step results in computing a product of a matrix with time-independent entries and a vector. The complexity of this operation is $O(N^2)$, assuming the matrix $\mathcal{A}$ is $N \times N$, and the vector is $N \times 1$. However, $N$ in this case is relatively small (see below). One can compare this with the FFT algorithm proposed in \cite{AA2000} to compute the correlation integral. This translates into computation of two matrix-by-vector products. This algorithm is $2 c \times O(N \log_2 N)$, where $c$ is some coefficient. However, $N$ is relatively high in this case. Typical values are $N = 4096$. Also a post-solution interpolation is required.\footnote{In more advanced approaches, this step could be eliminated; see \cite{Parrot2009}. Also if coefficients of the linear interpolation are pre-computed, overhead of performance for doing interpolation is relatively small, see \cite{Halluin2004,Halluin2005a}.} Finally, for some models (CGMY, VG), the computation of the integral in a neighborhood of $x=0$ requires special treatment (\cite{ContVolchkova2003}).

To make a numerical estimate, assume that we want to compute an option price on the grid with the local accuracy $O(h_i h_{i+1})$ where the option log-strike lies within the interval $[x_i, x_{i+1}]$ of the non-uniform grid ${\bf G}(x)$. Also suppose that the jump integral is truncated from the above at $x_{max} = \log S_{max} = 12.566$ (this approximately corresponds to a choice of the step size in the Fourier space $\eta=0.25$, see \cite{CarrMadan:98}).
Finally assume that the inverse Fourier transform integral is approximated by the Trapezoid rule, which provides same accuracy, e.g., $O(\lambda_1^2))$, where $\lambda_1$ is the step of integration in the log-strike space. Then to make a local error of both methods to be of the same order we need to set $\lambda_1 = 2 b/N$. For $h_i = 0.006$ this gives $N=4096$. On the other hand, a non-uniform grid ${\bf G}(x)$ for computing $J(\Delta_x)$ can be easily constructed, see e.g., \cite{HoutFoulon2010}, that has the same local step $h_i$ close to the strike, and ends up at $x_{max}$, while the total number of grid points $N$ is about 100-200. Therefore, in this case this method is able to outperform FFT. Further improvements, for instance, using the Simpson's rule for integration could be done in favor of the FFT approach. However, various non-uniform grids can also be used in our approach to reduce the number of nodes. Therefore, both methods seem to be comparable in performance. We demonstrate this below when presenting some numerical examples.

At the very least the product $\mathcal{A} C(x,\tau)$ can be computed also using FFT, if at every time step one re-interpolates values from ${\bf G}(x)$ to the FFT grid, similar to how this was done in \cite{Halluin2004}. The advantage of our method then is that it doesn't use Picard iterations to provide the second order approximation in space, which give some gain in performance as compared with the method of \cite{Halluin2004}. Also it is known that the latter method for the CGMY model experiences some problems when parameter $\alpha$ of the model is close to 2, while our method seems to be insensitive to that.

The above consideration is sufficiently general in the sense that it covers any particular jump model where jumps are modeled as an exponential \LY process. Clearly, as we already mentioned in Introduction, for some models computation of the jump integral can be readily simplified, for instance for the Merton's model, thus demonstrating a better performance than a more general approach.

\section{Examples for Some Popular Models}
In this section, we review some popular jump models known in the financial literature. Given a model, our goal is to construct a finite-difference scheme, first for $\Delta_x$, and then for $J(\Delta_x)$, that satisfies the conditions of Proposition~\ref{prop0}. We want to underline that we discuss these jump models being a part of a more general either LV or LSV model with jumps. Otherwise, as characteristic functions of the original Merton, Kou and CGMY models are known, any FFT based method would be more efficient in, e.g.,  obtaining prices of European vanilla options.

\subsection{Merton Model}
\cite{merton:76} considered jumps that are normally distributed with the \LY density
\begin{equation} \label{MertonDensity}
\nu(dx) = \lambda \frac{1}{\sqrt{2 \pi}\sigma_J} \exp \left[ -\frac{(x-\mu_J)^2}{2 \sigma^2_J}\right] dx,
\end{equation}
\noindent where $\lambda$, $\mu_J$ and $\sigma_J$ are parameters of the model. Considering the pure jump part of the Merton model, one can see that it exhibits finite activity, i.e., a finite number of jumps within any finite time interval. Plugging \eqref{MertonDensity} into the definition of the operator $\mathcal{J}$ in \eqref{intGen} and fulfilling a formal integration gives
\begin{equation} \label{JMerton}
\mathcal{J} = \lambda\left(e^{\mu_J \triangledown +\frac{1}{2} \sigma _J^2 \triangledown^2} -
\kappa \triangledown -1 \right), \qquad \kappa = e^{\mu _J+\frac{\sigma _J^2}{2}} - 1,
\end{equation}
\noindent where $\triangledown \equiv \partial /\partial x$,  $\triangledown^2 \equiv \partial^2 /\partial x^2$. The corresponding evolutionary pure jump equation to be solved is
\begin{equation} \label{EqMerton}
C^{(2)}(x,\tau) = \mathcal{A} C^{(1)}(x,\tau), \qquad
\mathcal{A} = \exp \left[ \lambda \Delta \tau \left(e^{\mu_J \triangledown +\frac{1}{2} \sigma _J^2
\triangledown^2} - \kappa \triangledown -1 \right) \right].
\end{equation}

A matrix exponential method for this model with the exponential operator\footnote{It is actually a double exponential operator.} as in \eqref{EqMerton} has already been  considered in \cite{Tangman2011} using a different derivation (from \cite{CarrMayo}). They also discuss in more detail various modern methods for computing the matrix exponentials.

Recall that the diffusion equations in \eqref{splitFin} have to be solved up to some order of approximation in time $\tau$. Suppose for this purpose we want to use a finite-difference scheme that provides a second order approximation, $O((\Delta \tau^2))$. However, \eqref{EqMerton} gives an {\it exact} solution of the corresponding pure jump equation (the second step in Strang's splitting scheme). Since Strang's scheme guarantees only second-order accuracy ($O((\Delta \tau)^2)$) to the exact solution of the full PIDE, the second step could be computed to the same order of accuracy.

To this end we can use the (1,1) P{\'a}de approximation of $e^{\Delta \tau \mathcal{J}}$,
\begin{equation} \label{mer1}
e^{\Delta \tau \mathcal{J}} \approx [1 - \frac{1}{2}\Delta \tau \mathcal{J}]^{-1}[1 + \frac{1}{2}\Delta \tau \mathcal{J}] + O(\Delta \tau^3).
\end{equation}
Now the product
\[ \mathcal{J} C^{(1)}(x,\tau) = - \lambda (\kappa \triangledown + 1) C^{(1)}(x,\tau) + \lambda e^{\mu_J \triangledown +\frac{1}{2} \sigma _J^2 \triangledown^2}
C^{(1)}(x,\tau) \]
can be efficiently computed if one observes that:
\begin{itemize}
\item Merton's jumps are that with finite variation and finite activity. Therefore, the term $- \lambda \kappa \triangledown C^{(1)}(x,\tau)$ could be taken out of the jump integral and added to the diffusion operator (see our splitting algorithm, \eqref{splitFin}). We will denote the remaining part of the integral as $\mathcal{J}^*$, e.g.,
\[ \mathcal{J}^* C^{(1)}(x,\tau) = \lambda \left[ -1 + e^{\mu_J \triangledown +\frac{1}{2} \sigma _J^2 \triangledown^2} \right] C^{(1)}(x,\tau) \]
\item Vector
\[ z(x,\tau) \equiv e^{\mu_J \triangledown +\frac{1}{2} \sigma _J^2 \triangledown^2} C^{(1)}(x,\tau) \]
 is a solution of
 \begin{equation} \label{heat}
\fp{z(x,s)}{s} = \left(\mu_J \triangledown +\frac{1}{2} \sigma _J^2 \triangledown^2\right) z(x,s).
\end{equation}
for $0 \le s \le 1$ and $z(x,0) = C^{(1)}(x,\tau)$. A straightforward approach proposed in \cite{CarrMayo} suggests to use, e.g., finite difference scheme to solve this equation. The solution should be obtained at the same grid in space with a space step $h$, while the "time" step $\Delta s$ could be arbitrary chosen. However, since the total accuracy of this solution should not be worse that the required accuracy of the whole method, e.g., $O(\Delta \tau^2 + h^2)$, this dictates that $\Delta s \le \max(h, \Delta \tau)$. Therefore, the total complexity of such the solution is $O(N M), \ M = 1/\Delta s$.

However, this result could be improved. Indeed, suppose we compute an European option price\footnote{The below approach is also applicable to single-barrier options, or to the options with a non-vanilla payoff, e.g., digitals.}. As coefficients  $\mu_J, \sigma _J$ are assumed to be constant, the Green's function of \eqref{heat} is Gaussian. Therefore, the solution of \eqref{heat} given a vector of the initial prices is a convolution of this vector with the Gaussian kernel, and it can be computed by using a Fast Gaussian Transform (FGT).

Since our problem is one-dimensional computation of the low-dimensional FGT does not pose
any difficulties if we use a powerful algorithm known as Improved Fast Gauss Transform (IFGT), see
\cite{IFGT}. The number of target points in this case is equal to the number of source points $N$, and, therefore, the total complexity of IFGT is $O(2N)$.

\item The scheme \eqref{fd0} with allowance for \eqref{mer1} can be re-written as
\[ C^{(1)}(x, \tau + \Delta \tau) - C^{(1)}(x, \tau) = \frac{1}{2} \Delta \tau \mathcal{J}^* \left[C^{(1)}(x, \tau + \Delta \tau) + C^{(1)}(x, \tau)\right], \]
\noindent and this equation could be solved using the Picard iterations having in mind that at each iteration vector $z(x,t)$ could be obtained by solving \eqref{heat}.
\end{itemize}
In other words, we presented another derivation of the method first proposed in \cite{CarrMayo}\footnote{With a proposed improvement that reduces the total complexity of the method from $O(N /\Delta s)$ to $O(N)$.}. Notice, that to be unconditionally stable this method requires $B \equiv \mu_J \triangledown +\frac{1}{2} \sigma _J^2 \triangledown^2$ to be a Metzler matrix. Then $e^B$ is a positive matrix with all positive eigenvalues less than 1 in value. Accordingly,  $ \mathcal{J}^* = \lambda(-I + e^B)$ is a Metzler matrix with all negative eigenvalues. Then $|B_1^{-1} B_2| < 1$, where $B_1 = I + \frac{1}{2} \Delta \tau \mathcal{J}^*,  \ B_2 = I - \frac{1}{2} \Delta \tau \mathcal{J}^*$, and $I$ is an identity matrix.

\subsection{Kou Model}
The Kou model, proposed in \cite{Kou2004}, is a double exponential jump model. Its \LY density is
\begin{equation} \label{Kou}
\nu (dx) = \lambda\left[ p \theta_1 e^{-\theta_1 x} {\bf 1}_{x \ge 0} + (1-p) \theta_2 e^{\theta_2 x} {\bf 1}_{x < 0} \right] dx,
\end{equation}
where $\theta_1 > 1$, $\theta_2 > 0$, $1 > p > 0$; the first condition was imposed to ensure that the stock price $S(t)$ has finite expectation. Using this density in the definition of the operator $\mathcal{J}$ in \eqref{intGen} and carrying out the integration (recalling that we treat $\partial/ \partial x$ as a constant) gives
\begin{align} \label{KouJ}
\mathcal{J} &= \lambda  \left[- 1 + \mu_0 \triangledown + p \theta_1(\theta_1-\triangledown)^{-1} + (1-p) \theta_2(\triangledown+\theta_2)^{-1}\right], \\
\triangledown &\equiv \partial_x, \qquad \mu_0 = \frac{p}{\theta _1-1} - \frac{1-p}{1+\theta _2}, \quad -\theta_2 <  Re(\triangledown) < \theta_1. \nonumber
\end{align}
The inequality $-\theta_2 <  Re(\triangledown) < \theta_1$ is an existence condition for the integral defining $\mathcal{J}$ and should be treated as follows: the discretization of the operator $\triangledown$ should be such that all eigenvalues of matrix $A$, a discrete analog of $\triangledown$, obey this condition.

Also for the future let us remind that $\lambda$ is a parameter (intensity) of the Poison process, therefore $\lambda > 0$.

We proceed in a similar to Merton's model way by using again the (1,1) P{\'a}de approximation of $e^{\Delta \tau \mathcal{J}}$. As Kou's jumps are that with finite variation and finite activity, the term $\lambda  \mu_0 \triangledown$ could be taken out of the jump integral and added to the diffusion operator (see our splitting algorithm, \eqref{splitFin}). Now the whole product $\mathcal{J}^* C^{(1)}(x,\tau)$ with \[ \mathcal{J}^*  \equiv -1+ p \theta_1 (\theta_1-\triangledown)^{-1} + (1-p)  \theta_2(\triangledown+\theta_2)^{-1} \]
\noindent could be calculated as follows.
\paragraph{Second term.}
Observe that vector $z(x,\tau) = p \theta_1 (\theta_1-\triangledown)^{-1} C^{(1)}(x,\tau)$ solves the equation
\begin{equation} \label{kou1system}
(\theta_1-\triangledown) z(x,\tau) = p \theta_1 C^{(1)}(x,\tau)
\end{equation}
The lhs of this equation could be approximated to $O(h^2)$ using a {\it forward} one-sided derivative
$\triangledown f(x) = - [3 f(x) - 4 f(x+h) + f(x + 2 h)]/(2 h) + O(h^2)$, so on a given grid matrix $A^F_2$ with elements $-3/(2h)$ on the main diagonal, $2/h$ on the first upper diagonal, and $-1/(2h)$ on the second upper diagonal is a representation of $\triangledown$. Note, that the matrix $M_1 = \theta_1 I -A^F_2$ is not an M-matrix, however its inverse is a positive matrix if $h < 1/\theta_1$. Also since $M_1$ is an upper banded tridiagonal matrix, its eigenvalues are $\lambda_i = \theta_1 + 3/(2h), i=1,N$. Also under the condition $h < 1/\theta_1$ one has $|p \theta_1/\lambda_i| < 1$, i.e. this discretization is unconditionally stable in $h$ given the above condition is valid. Solving \eqref{kou1system} vector $z(x,\tau)$ can be found with the complexity $O(N)$.
\paragraph{Third term.}
Observe that vector $z(x,\tau) = (1-p) \theta_2 (\theta_2+\triangledown)^{-1} C^{(1)}(x,\tau)$ solves the equation
\begin{equation} \label{kou2system}
(\theta_2+\triangledown) z(x,\tau) = (1-p) \theta_2 C^{(1)}(x,\tau)
\end{equation}
The lhs of this equation could be approximated with $O(h^2)$ using a {\it backward} one-sided derivative
$\triangledown f(x) = [3 f(x) - 4 f(x-h) + f(x - 2 h)]/(2 h) + O(h^2)$, so on a given grid matrix $A^B_2$ with elements $3/(2h)$ on the main diagonal, $-2/h$ on the first lower diagonal, and $1/(2h)$ on the second lower diagonal is a representation of $\triangledown$. Note, that the matrix $M_2 = \theta_2 I + A^B_2$ is not an M-matrix, however its inverse is a positive matrix if $h < 1/\theta_2$. Also since $M_2$ is an lower banded tridiagonal matrix, its eigenvalues are $\lambda_i = \theta_2 + 3/(2h), i=1,N$. Also under the condition $h < 1/\theta_2$ one has $|(1-p)\theta_2/\lambda_i| < 1$, i.e. this discretization is unconditionally stable in $h$ given the above condition is valid. Solving \eqref{kou2system} vector $z(x,\tau)$ can be found with the complexity $O(N)$.

Overall, a discrete representation of $\mathcal{J}^*$ on the given grid constructed in such a way is a Metzler matrix, therefore all its eigenvalues have a negative real part. Indeed, all eigenvalues of the matrix $M_1^{-1}$ (here they are just the diagonal elements) are positive and less than 1, and all eigenvalues of the matrix $M_2^{-1}$ (also here they are just the diagonal elements) are positive and less than 1. Moreover, their sum is less than 1, and, therefore, the diagonal elements of matrix $\mathcal{J}^*$ are negative and less than 1.

Now by construction, it could be seen that matrices $M_1, M_2$ are strictly diagonal dominant, and, therefore, the off-diagonal elements of matrices $M_1^{-1}, M_2^{-1}$ are small as compared with that on the main diagonal. Therefore, by Gershgorin's circle theorem (\cite{GL83}) eigenvalues of $\mathcal{J}^* $ are $|\lambda_i| < 1, \ i=1,N$. Thus, the above described scheme is unconditionally stable provided $h < 1/\max(\theta_1, \theta_2)$, and at the same time gives the second order approximation $O(h^2 + \Delta \tau^2)$.

\paragraph{Numerical experiments}
Note, that aside of splitting technique and the way how to solve the diffusion equations at the first and third steps of Strang's splitting, our method differs from that in \cite{Halluin2004} only by how we compute a jump integral. Therefore, our numerical experiments aim to compare just that part and are organized as follows.

We consider a call option and take the Kou model parameters similar to \cite{Halluin2005b}, i.e., $S_0 = K = 100, r = 0.05, p = 0.0.3445, \theta_1 = 3.0465, \theta_2 = 3.0775, \sigma=0.15$. One step in time is computed by taking $T = \Delta \tau = 0.25$ (same as in \cite{Halluin2005b}). As $C^{(1)}(x, \Delta \tau)$ in the \eqref{splitFin} comes after the first step of splitting, we get it by using the Black-Scholes formula with the forward interest rate $r + \lambda \mu_0$ because in our splitting algorithm we moved the term $\lambda \mu_0 \triangledown$ from the jump part to the diffusion part (see above). At the second step the solution of the jump part $C^{(2)}_j(x,\Delta \tau)$ is produced given the initial condition $C^{(1)}(x,\Delta \tau)$ from the previous step. We compare our solution for the jump step with that obtained with $N = 409601$ which is assumed to be close to the exact value\footnote{This method is not very accurate. But as the exact solution is not known, it provides a plausible estimate of the convergence.}. The finite-difference grid was constructed as follows: the diffusion grid was taken from $x^D_{min} = 10^{-3}$ to $x^D_{max} = 30 max(S,K)$. The jump grid is a superset of the diffusion grid, i.e. it coincides with the diffusion grid at the diffusion domain and then extends this domain up to $x^J_{max} = \log (10^5)$. Here to simplify the convergence analysis we use an uniform grid with step $h$. However, non-uniform grid can be easily constructed as well, and, moreover, that is exactly what this algorithm was constructed for.

The results of such a test are given in Table~\ref{Tab1}. Here $C$ is the price in dollars, $N$ is the number of grid nodes, $t_e$ is the elapsed time\footnote{All experiments were computed in Matlab at Intel Pentium 4 CPU 3.2 Ghz under x86 Windows 7 OS. Obviously, C++ implementation provides a better performance by roughly factor 5.}, $\beta$ is the order of convergence of the scheme. The "exact" price obtained at $N= = 409601$ is $C_{num}(\Delta \tau)$ =  3.99544616155. It is seen that the convergence order $\beta_i
= \log_2 \frac{C(i) - C_{num}}{C(i+1)-C_{num}}, \ i=1,2...$ of the scheme is asymptotically close to $O(h^2)$.
\begin{table}[h!]
\begin{center}
\begin{tabular}{|c|l|r|l|r|}
\hline
$C$ & $h$ & $N$ & $t_e, \mbox{sec}$ & $\beta$ \cr
\hline
 4.08176114 & 0.149141 & 101   & 0.00807 & - \cr
 \hline
  4.00896884 & 0.0745706 & 201 & 0.00441 & 3.81602 \cr
  \hline
  3.99640628 & 0.0372853 & 401 & 0.00613 & 2.02002 \cr
  \hline
  3.99568288 & 0.0186427 & 801 & 0.00772 & 2.04431 \cr
  \hline
  3.99550355 & 0.00932133 & 1601 & 0.00829 & 1.93623 \cr
  \hline
  3.99546116 & 0.00466066 & 3201 & 0.01042 & 1.96681 \cr
  \hline
  3.99545000 & 0.00233033 & 6401 & 0.02305 & 1.97418 \cr
\hline
  3.99544714 & 0.00116517 & 12801 & 0.04445 & 1.97265 \cr
  \hline
  3.99544641 & 0.000582583 & 25601 & 0.10146 & 1.96645 \cr
\hline
  3.99544623 & 0.000291291 & 51201 & 0.22158 & 1.99185 \cr
\hline
  3.99544618 & 0.000145646 & 102401 & 0.53087 & 2.21324 \cr
\hline
\end{tabular}
\caption{Convergence of the proposed scheme for Kou's model, $T = \Delta \tau = 0.25$}
\label{Tab1}
\end{center}
\end{table}

As a sanity check we can compare this value with the reference value obtained by pricing this model (one step) using FFT, which is $C_{FFT}(\Delta \tau)$ = 3.97383, see, e.g., \cite{Halluin2005b}. Definitely $C_{FFT}(\Delta \tau)$ is not exactly equal to $C_{num}(\Delta \tau)$ because our two steps used in the test\footnote{Don't miss this with the accuracy of the whole 3 steps Strang's algorithm which is $O(\Delta \tau^2)$. The test validates just the convergence in $h$, not in $\Delta \tau$.} are equivalent to the splitting scheme of the first order in $\Delta \tau$, i.e. it has an error $O(\Delta \tau)$. And $\Delta \tau$ in this experiment is large. Therefore, we rerun this test taking now $T = \Delta \tau = 0.5$. This results are given in Tab.~\ref{Tab11}. Now $C_{FFT}(\Delta \tau)$ =  1.545675, and $C_{num}(\Delta \tau)$ = 1.544557, so the relative error is 0.07\%. This confirms that the value $C_{num}(\Delta \tau)$ looks reasonable.

\begin{table}[h!]
\begin{center}
\begin{tabular}{|c|l|r|l|r|}
\hline
$C$ & $h$ & $N$ & $t_e, \mbox{sec}$ & $\beta$ \cr
\hline
  1.96362542 & 0.149141 & 101 & 0.00819 & - \cr
\hline
  1.72184850 & 0.0745706 & 201 & 0.00387 & 5.00130 \cr
\hline
  1.55009251 & 0.0372853 & 401 & 0.00692 & 3.62747 \cr
\hline
  1.54500503 & 0.0186427 & 801 & 0.01647 & 4.78856 \cr
\hline
  1.54457335 & 0.00932133 & 1601 & 0.01135 & 1.94876 \cr
\hline
  1.54456134 & 0.00466066 & 3201 & 0.01330 & 2.05321 \cr
\hline
  1.54455816 & 0.00233033 & 6401 & 0.02666 & 2.07878 \cr
\hline
  1.54455739 & 0.00116517 & 12801 & 0.04160 & 1.98360 \cr
\hline
  1.54455721 & 0.000582583 & 25601 & 0.10207 & 1.99972 \cr
\hline
  1.54455716 & 0.000291291 & 51201 & 0.22687 & 2.05112 \cr
\hline
  1.54455715 & 0.000145646 & 102401 & 0.77614 & 2.29722 \cr
\hline
\end{tabular}
\caption{Convergence of the proposed scheme for Kou's model, $T = \Delta \tau = 0.05$.}
\label{Tab11}
\end{center}
\end{table}

Performance-wise the similarity of this method to that in \cite{Halluin2004} is that it also requires Picard's iterations at every time step. In contrast to \cite{Halluin2004} at every iteration this method requires solution of two linear systems with a tridiagonal (one upper and one lower triangular) matrix, i.e. its complexity is $O(N)$. In \cite{Halluin2004} it requires two FFT provided on a slightly extended grid to avoid wrap-around effects, so the total complexity is at least $O(N \log_2 N)$. Therefore, even if $N$ in out method is chosen to be close to $N$ in the FFT approach the former is approximately $\log_2 N$ times faster.

\subsection{CGMY Model}
Computation of jump integrals under the CGMY model (also known as the KoBoL model, or more generally as generalized tempered stable processes (GTSPs)) was considered in detail in \cite{ItkinCarr2012Kinky} using a similar approach. GTSPs have probability densities symmetric in a neighborhood of the origin and exponentially decaying in the far tails. After this exponential softening, the small jumps keep their initial stable-like behavior, whereas the large jumps become exponentially tempered. The L\'evy measure of GTSPs is given by
\begin{equation}\label{measure}
    \mu(y) = \lambda_{L} \dfrac{e^{-\nu_{L}|y|}}{|y|^{1 + \alpha_{L}}}{\bf 1}_{y<0} + \lambda_{R} \dfrac{e^{-\nu_{R}|y|}}{|y|^{1 + \alpha_{R}}}{\bf 1}_{y>0},
\end{equation}
where $\nu_R$, $\nu_L > 0$, $\lambda_R$, $\lambda_L > 0$ and $\alpha_R, \alpha_L < 2$. The last condition is necessary to provide
\begin{equation}\label{cond}
    \int^1_{-1} y^2 \mu(dy) < \infty \, , \ \int_{|y| > 1} \mu(dy) < \infty.
\end{equation}

The next proposition follows directly from Proposition 7 of \cite{ItkinCarr2012Kinky}.\footnote{In Itkin and Carr's paper, jump integrals were defined on half-infinite positive and negative domains, while here they are defined on the whole infinite domain. Therefore, to prove this Proposition simply use $\int_{-\infty}^{\infty} = \int_{-\infty}^{0} + \int_{0}^{\infty}$ and then apply Proposition 7 from  \cite{ItkinCarr2012Kinky}}
\begin{proposition}
The PIDE
\begin{equation*}
\fp{}{\tau} C(x,\tau) = \int_{-\infty}^{\infty} \left[C(x+y,\tau) - C(x,\tau) - \fp{}{x} C(x,\tau) (e^y-1)  \right] \mu(y) dy
\end{equation*}
is equivalent to the PDE
\begin{align} \label{whole}
\fp{}{\tau} C(x,\tau) &= (\mathcal{L}_R + \mathcal{L}_L)C(x,\tau), \\
\mathcal{L}_R &= \lambda_R \Gamma(-\alpha_R) \left\{ \left(\nu_R - \triangledown\right)^{\alpha_R} - \nu_R^{\alpha_R} + \left[ \nu_R^{\alpha_R} - (\nu_R-1)^{\alpha_R}\right] \triangledown \right\}, \nn \\
& \alpha_R < 2, \ {Re}(\nu_R - \triangledown) > 0, \, \nu_R > 1, \nn \\
\mathcal{L}_L &= \lambda_L \Gamma(-\alpha_L) \left\{ \left(\nu_L + \triangledown\right)^{\alpha_L} - \nu_L^{\alpha_L} + \left[ \nu_L^{\alpha_L} - (\nu_L+1)^{\alpha_L}\right] \triangledown \right\}, \nn \\
& \alpha_L < 2, \ {Re}(\nu_L + \triangledown) > 0, \ \nu_L > 0, \nn
\end{align}
where $\Gamma$ is the gamma function, and $Re(L)$ for some operator $L$ formally refers to the spectrum of $L$. In other words, $Re(L) > 0$  means that real parts of all eigenvalues $\lambda$ of $L$ are positive.

In special cases, this equation changes to
\begin{align} \label{whole0}
\mathcal{L}_R  &= \lambda_R  \left\{ \log(\nu_R) - \log \left(\nu_R - \triangledown \right) + \log \left(\dfrac{\nu_R-1}{\nu_R}\right)\triangledown  \right\} \\
& \alpha_R = 0, \mathbb{R}(\nu_R - \triangledown) > 0, \mathbb{R}(\nu_R) > 1, \nn \\
\mathcal{L}_L  &= \lambda_L  \left\{\log(\nu_L)  - \log \left(\nu_L + \triangledown \right)  + \log \left(\dfrac{\nu_L+1}{\nu_L}\right)\triangledown  \right\} \nn \\
& \alpha_L = 0, \ \mathbb{R}(\nu_L + \triangledown) > 0, \ \mathbb{R}(\nu_L) > 0, \nn
\end{align}
and
\begin{align} \label{whole1}
\mathcal{L}_R  &=  \lambda_R \Big[ (\nu_R-\triangledown)\log (\nu_R-\triangledown) - \nu_R \log (\nu_R) + \triangledown \left(\log (\nu _R-1) - 2 \nu_R \coth^{-1} (1-2 \nu_R) \right)  \Big]  \nn \\
& \alpha_R = 1, \ Re(\nu_R - \triangledown) > 0, \ \nu_R > 1, \\
\mathcal{L}_L  &= \lambda_L \Big[ (\nu_L+\triangledown)\log \left(\frac{\nu_L+\triangledown}{\nu_L}\right) - \triangledown (1 + \nu_L)\log \left(\frac{\nu_L+1}{\nu_L}\right) \Big]  \nn \\
& \alpha_L = 1, \ Re(\nu_L + \triangledown) > 0, \ \nu_L > 0, \nn
\end{align}
where the logarithm of the differential operator is defined in the sense of \cite{logOfDif}.
\end{proposition}

We underline the existence conditions for the jump integrals to be well-defined which are
$\nu_L > 0, \ \nu_R > 1$. This is in some sense similar to the Kou's model where $\theta_1$ is defined on the domain $\theta_1 > 1$ while $\theta_2$ at the domain $\theta_2 > 0$.

There are a few ways to proceed in this case. First, one can use an extra Strang's splitting; instead of directly solving \eqref{whole}, solve it in three sweeps. At every step, only one operator, either  $\mathcal{L}_R$ or $\mathcal{L}_L$ enters the equation.  Thus, the construction of the appropriate discrete operator is simplified. The second approach is based on the observation that eigenvalues of the sum of two M-matrices are also positive. This result follows from Wayl's inequality (see \cite{Bellman}). Therefore, if every operator in the right-hand side of \eqref{whole} is represented by the negative of an M-matrix, the sum of those operators is also the negative of an M-matrix. However, the discretization of these operators, while on the same grid, could differ, thus adding some flexibility to the construction of the numerical scheme.

As shown in \cite{ItkinCarr2012Kinky}, the computation of the matrix exponential could be fully eliminated by using the following approach. First, they show that for $\alpha_I \in \mathbb{Z}$ the solution of the pure jump equation could be reduced to the solution of a system of linear equations where the matrix in the left-hand side of the system is banded. Therefore, the complexity of this solution is $O(N)$. Then to compute the matrix exponential for a real $\alpha$, first choose three closest values of $\alpha_I \in \mathbb{Z}$. Given the solutions at these $\alpha_I$, we can interpolate them to give the solution for $\alpha$. Therefore, if linear interpolation is used, and the interpolation coefficients are pre-computed, the total complexity of this solution is also $O(N)$.

This approach, however, does not work well if $0 < \alpha < 2$, since we do not have a solution at $\alpha = 2$.  To proceed in such a way would then require extrapolation  instead of interpolation. It is well known that extrapolation is not a reliable procedure, and so in what follows we apply the general approach of this paper to the GTSP models.

First, consider terms with $\alpha_R$. Based on the above analysis, the most important case for us is $1 < \alpha_R < 2$. That is because if we manage to propose an efficient numerical algorithm in such case, other domains of $\alpha_R$ could be treated as in \cite{ItkinCarr2012Kinky} by involving the value $1 < \alpha_R < 2$ into the interpolation procedure in \cite{ItkinCarr2012Kinky}. However, for the sake of completeness we begin with a relatively simple case $\alpha_R < 0$ and $0 < \alpha_R < 1$ to demonstrate our approach. A special case $\alpha_R = 0$ was already addressed in \cite{ItkinCarr2012Kinky}. A special case
$\alpha_R = 1$ is considered later in this paper.

\subsubsection{Case $\alpha_R < 0$.}
Define a one-sided {\it forward} discretization of $\triangledown$, which we denote as $A^F: \ \partial C/ \partial x = [C(x+h,t) - C(x,t)]/h$. Also define a one-sided {\it backward} discretization of $\triangledown$, denoted as $A^B: \ \partial C/ \partial x = [C(x,t) - C(x-h,t)]/h$.
\begin{proposition} \label{prop-0}
If $\alpha_R < 0$, then the discrete counterpart $L_R$ of the operator $\mathcal{L}_R$ is the negative of an M-matrix if
\[ L_R = \lambda_R \Gamma(-\alpha_R) \left\{ \left(\nu_R I - A^F\right)^{\alpha_R} - \nu_R^{\alpha_R} I + \left[ \nu_R^{\alpha_R} - (\nu_R-1)^{\alpha_R}\right] A^B \right\}.\]
The matrix $L_R$ is an $O(h)$ approximation  of the operator $\mathcal{L}_R$.
\end{proposition}
\begin{proof}
We need eigenvalues of $A^F$ to be negative to obey the existence condition in \eqref{whole}. That dictates the choice of $A^F$ in the first term as $A^F$ is the Metzler matrix which eigenvalues are negative.
Now take into account that $\nu_R^{\alpha_R} - (\nu_R-1)^{\alpha_R} < 0$ if $\alpha_R < 0$,
while $\Gamma(-\alpha_R) > 0$. Matrix $M = \left(\nu_R I - A^F\right)$ is an M-matrix with all positive eigenvalues. Its power is a positive matrix because $M^{\alpha_R} = \exp(\alpha_R \log M)$, matrix $\log M$ is also an M-matrix, matrix $\alpha_R \log M$ is negative of an M-matrix, i.e. the Metzler matrix, and exponentiation of the Metzler matrix gives a positive matrix, see \cite{BermanPlemmons1994}). Matrix $M_1 =   - \nu_R^{\alpha_R} + \left[ \nu_R^{\alpha_R} - (\nu_R-1)^{\alpha_R}\right] A^B $ is bi-diagonal and also the Metzler matrix. Therefore, $M+M_1$ is the Metzler matrix, so is $L_R$. Now take into account that diagonal elements of $M$ are $d_i < (\nu_R + 1/h)^{\alpha_R}, \ i=1,N$, and diagonal elements of $M_1$ are
$d_{1,i} = [\nu_R^{\alpha_R} - (\nu_R-1)^{\alpha_R}]/h - \nu_R^{\alpha_R}, \ i=1,N$. Therefore,
\[ d_i + d_{1,i} <  \left(\nu_R + \frac{1}{h}\right)^{\alpha_R} + \frac{1}{h}[\nu_R^{\alpha_R} - (\nu_R-1)^{\alpha_R}] - \nu_R^{\alpha_R} <  \frac{1}{h}[\nu_R^{\alpha_R} - (\nu_R-1)^{\alpha_R}] < 0 \]
Thus, matrix $L_R$ is the negative of an M-matrix. First order approximation follows from the definition of $A^F$ and $A^B$. $\blacksquare$
\end{proof}

To get the second order of approximation we can use the following observations:
\begin{itemize}
\item Jumps with $\alpha_R < 0$ are of the finite activity and finite variation. Therefore, the term
$\left[ \nu_R^{\alpha_R} - (\nu_R-1)^{\alpha_R}\right] \triangledown$ could be moved to the diffusion part of our splitting algorithm;
\item The remaining operator could be approximated as
\[ L_R = \lambda_R \Gamma(-\alpha_R) \left\{ \left(\nu_R I - A^F_2\right)^{\alpha_R} - \nu_R^{\alpha_R}I \right\} + O(h^2) \]
\end{itemize}
The proof is almost exactly same as in the proposition~\ref{prop-0}, if one notices that despite
$M = \nu_R I - A^F_2$ is not exactly an M-matrix, its logarithm is an M-matrix. That is because $M$ is upper tridiagonal matrix which positive elements on the second upper diagonal in absolute value are small as compared with the elements of the main and first upper diagonals.

\paragraph{Numerical experiments} We organize this test in exactly same way as that was done for the Kou model. There are two ways to proceed. The first one is to pre-compute $\mathcal{J} = \exp(\Delta \tau L_R)$, and then at every time step of the splitting method when a corresponding jump equation has to be solved (or a jump integral has to be computed) to compute a product $\mathcal{J} C(x,\tau)$. This operation has the complexity $O(N^2)$, but it doesn't require Picard iterations to provide the second order approximation in $\tau$. Another approach would be to proceed in a sense of \cite{Halluin2004}, similar to what we did for the Merton and Kou models. The (1,1) P{\'a}de approximation of $e^{\Delta \tau \mathcal{J}}$ could be again re-written in the form of the implicit equation which could be solved by using Picard iterations (see above). Here, however, we don't have a fast way to compute a product $\mathcal{J} C(x,\tau)$, so FFT could be used for this purpose. From this prospective, this method is similar to \cite{Halluin2004}, the difference is in the matrix $L_R$. We, however, remind the reader, that the method of \cite{ItkinCarr2012Kinky} is more efficient in this case.

In Tab.~\ref{tab2} the results of such a test for a call option are given assuming the following values of parameters: $\alpha_R = -0.5, \lambda_R = 10, \nu_R = 2, S_0 = K = 1, r = 0, \sigma = 0.2, T = 0.1$. The grid was constructed exactly in the same way as in the test for Kou's model.
in Table~\ref{tab2} $C_{it}$ is the price in cents obtained by using Picard iterations,
$C_{exp}$ is the price in cents obtained by using matrix exponential,
$N$ is the number of grid nodes, $\beta_{it}$ is the order of convergence of the iterative scheme,
$\beta_{exp}$ is the order of convergence of the exponential scheme. The "exact" price obtained at $N=4000$ is $C_{it}(\Delta \tau) = $ 40.2261 cents, and $C_{exp}(\Delta \tau) = $ 39.223 cents. It is seen that the convergence order $\beta$ of both schemes is close to $O(h^2)$.

\begin{table}[h!]
\begin{center}
\begin{tabular}{|c|c|c|c|c|c|}
\hline
$C_{it}$ & $h$ & $N$ & $\beta_{it}$ & $C_{exp}$ & $\beta_{exp}$ \cr
\hline
40.1100 & 0.104131 & 100 & - & 39.1027 & - \cr
\hline
40.2002 & 0.051804 & 200 & 2.16 & 39.1937 & 2.03\cr
\hline
40.2223 & 0.025837 & 400 & 2.86 & 39.2167 & 2.26\cr
\hline
40.2260 & 0.012902 & 800 & 4.12 & 39.2216 & 2.44\cr
\hline
40.2258 & 0.006447 & 1600 & 1.58 & 39.2222 & 1.14 \cr
\hline
\end{tabular}
\caption{Convergence of the proposed scheme for CGMY model with $\alpha_R = -0.5$.}
\label{tab2}
\end{center}
\end{table}

At high $N$ the convergence ratio drops down most likely because computation of the matrix exponent, or matrix power loses accuracy, see \cite{Moler2003}.

\subsubsection{Case $0 < \alpha_R < 1$.}
This case is similar to the previous one.
\begin{proposition} \label{alpha01}
Suppose $0 < \alpha_R < 1$ (so jumps are of the infinite activity but finite variation) and consider the following discrete approximation of the operator  $\mathcal{L}_R$:
\[\lambda_R \Gamma(-\alpha_R) \left\{ \left(\nu_R I - A^F\right)^{\alpha_R} - \nu_R^{\alpha_R}I + \left[ \nu_R^{\alpha_R} - (\nu_R-1)^{\alpha_R}\right] \triangledown \right\}.\]
Because of the finite variation of the jumps the last terms in this representation could be taken out and moved to the diffusion part (that is what we did already several times in the above). The remaining matrix
\[ L_R = \lambda_R \Gamma(-\alpha_R) \left\{ \left(\nu_R I - A^F\right)^{\alpha_R} - \nu_R^{\alpha_R} I \right\}.\]
\noindent approximates the operator $\mathcal{L}_R$ with $O(h)$, and is the negative of an M-matrix.
\end{proposition}
The proof is also similar. The difference in the proof is as follows: $0 < \alpha_R < 1$ means
that $\Gamma(\alpha_R) < 0$. As $\alpha_R > 0$ matrix $M^{\alpha_R} = \exp(\alpha_R \log M)$
is an M- matrix, so is $M_1 = \left(\nu_R I - A^F\right)^{\alpha_R} - \nu_R^{\alpha_R}$.
The last statement is true because matrix $M$ is upper bi-diagonal, therefore $M^{\alpha_R}$ is upper triangular with diagonal elements $d_i = \left(\nu_R + 1/h\right)^{\alpha_R}$ ( this follows from the definition of the matrix power via a spectral decomposition). As $\left(\nu_R I + 1/h\right)^{\alpha_R}  - \nu_R^{\alpha_R} > 0$ diagonal elements of $M_1$ are positive. Thus, $M_1$ is an M-matrix, and $L_R$ is the negative of an M-matrix.

We run another test with the model parameters same as in the previous one and $\alpha_R = 0.9$. The results are given in Tab.~\ref{tab3}. One can observe the first order convergence in $h$. The "exact" price is $C_{exp}=$ 22.27 cents.
\begin{table}[h!]
\begin{center}
\begin{tabular}{|c|c|c|c|}
\hline
$C_{exp}$ & $h$ & $N$ & $\beta_{exp}$ \cr
\hline
23.9336 & 0.104131 & 100 & - \cr
\hline
22.9222 & 0.051804 & 200 & 1.35 \cr
\hline
22.5558 & 0.025837 & 400 & 1.19 \cr
\hline
22.3944 & 0.012902 & 800 & 1.21 \cr
\hline
22.3170 & 0.006447 & 1600 & 1.43 \cr
\hline
22.2789 & 0.003223 & 3200 & 2.59 \cr
\hline
\end{tabular}
\caption{Convergence of the proposed scheme for CGMY model with $\alpha_R = 0.9$.}
\label{tab3}
\end{center}
\end{table}

However, the second order approximation $O(h^2)$ cannot be constructed by simply replacing $A^F$ with $A^F_2$ when $\alpha_R$ is close to 1. For now we leave this as an open problem. As a work-around, in the next section $O(h^2 + \Delta \tau^2)$ algorithm is  constructed for $1 < \alpha_R < 2$. Then using a price obtained for some $\alpha_R^*, \ 1 < \alpha_R^* < 2$ and prices for $\alpha_R = 0, -1$ obtained using
the approach of \cite{ItkinCarr2012Kinky} (the latter could be computed with the complexity $O(N)$)
an $O(h^2)$ approximation for $ 0 < \alpha_R < 1$ can be found by interpolation.

\subsubsection{Case $\alpha_R = 1$.}
This case could be covered twofold. First, if we have a good method for the region $1 < \alpha_1 < 2$, then prices at $\alpha = 1$ could be obtained by computing three prices at $1 < \alpha_1 < 2, \ \alpha_2 \le 0, \alpha_3 < 0$ and then using interpolation in $\alpha$. This approach relies on the fact that for the CGMY model jump integrals are continuous in $\alpha$ at $\alpha < 2$, see Proposition 5 in \cite{ItkinCarr2012Kinky}.

Another approach is very similar to the previous case $0 < \alpha_R < 1$.
\begin{proposition} \label{alpha1}
Suppose $\alpha_R = 1$ and consider the following discrete approximation of $\mathcal{L}_R$:
\[
L = \lambda_R \Big[ (\nu_R-A^F)\log (\nu_R-A^F) - \nu_R \log (\nu_R)I + \kappa A^F\left(\log (\nu _R-1) - 2 \nu_R \coth^{-1} (1-2 \nu_R) \right)  \Big].
\]
\noindent where $\kappa$ is some constant. This approximates the operator $\mathcal{L}_R$ with $O(h)$, and is the negative of an M-matrix.
\end{proposition}
The proof is also similar. Indeed, according to this discretization $M_1 = \nu_R-A^F$ is an M-matrix, therefore $\log M_1$ is also an M-matrix. The product of this two M-matrices is an upper triangular matrix with all positive elements except of that at the first upper diagonal. Now observe that $\log (\nu _R-1) - 2 \nu_R \coth^{-1} (1-2 \nu_R) > 0$. Therefore, taking $\kappa > 0$ big enough dumps the negative values at the first upper diagonal and at the same time makes elements of the main diagonal all negative. Thus, the whole matrix $L$ is the negative of an M-matrix.

As in the original jump integral we have just $\kappa = 1$ the trick is to borrow $\Delta D = (\kappa-1)\left[\log (\nu _R-1) - 2 \nu_R \coth^{-1} (1-2 \nu_R) \right]\triangledown$ term from the diffusion part. In other words, we can re-distribute some terms in our splitting algorithm between the
diffusion and jump parts, as we did that for Kou and Merton models, and for CGMY
model with $\alpha_R < 1$, by moving a drift-like term from the diffusion to the jump part.
Accordingly, to compensate we need to subtract $\Delta D$ from the drift term in the diffusion part.
This potentially could result in the negative drift term which, however, is not a problem.

The results (call option prices in dollars) given below in Tab.~\ref{TabAL1} are obtained by applying this algorithm in the test with parameters $S = K = 100, T = 0.05; r = 0.05, \sigma = 0.15, \lambda_R = 0.1, \nu_R = 2, \kappa = 5$. The exact price is $C = 2.1428$ cents was obtained at $N = 2000$.
\begin{table}[h!]
\begin{center}
\begin{tabular}{|c|c|c|c|}
\hline
$C_{exp}$ & $h$ & $N$ & $\beta_{exp}$ \cr
\hline
  3.7296 & 0.1381550 & 101 &  - \cr
\hline
  2.6527 & 0.0690776 & 201 & 1.638 \cr
\hline
  2.3939 & 0.0345388 & 401 & 1.022 \cr
\hline
  2.2402 & 0.0172694 & 801 & 1.366 \cr
\hline
  2.1594 & 0.0086347 & 1601& 2.552\cr
\hline
\end{tabular}
\caption{Convergence of the proposed scheme for CGMY model with $\alpha_R = 1$.}
\label{TabAL1}
\end{center}
\end{table}
The first order convergence could be observed.

Similar to the previous case the second order approximation $O(h^2)$ cannot be constructed by simply replacing $A^F$ with $A^F_2$. We leave this as an open problem as well. At the beginning of this section we mentioned an interpolation approach which is applicable if the second order approximation could be constructed for $1 < \alpha_R < 2$. Then it could be used as a work-around to construct the $O(h^2)$ approximation.

\subsubsection{Case $1 < \alpha_R < 2$.}

This case is the most difficult, see, e.g., \cite{WangWanForsyth2007}. Below based on our general approach we provide an analysis of why a standard method experiences a problem in this range of the $\alpha_R$ values, and describe a variation of our method to address this problem.

Consider a discrete counterpart $L_R$ of the operator $\mathcal{L}_R$
\begin{equation} \label{cgmy12}
 L_R = \lambda_R \Gamma(-\alpha_R) \left\{ \left(\nu_R I - A_1\right)^{\alpha_R} - \nu_R^{\alpha_R}I + \left[ \nu_R^{\alpha_R} - (\nu_R-1)^{\alpha_R}\right] A_2 \right\}.
\end{equation}
\noindent where $A_1, A_2$ are some discrete approximations of the operator $\triangledown$ (i.e. $A_1 \propto 1/h, \ A_2 \propto 1/h$). Observe, that for this range of $\alpha_R$ the following inequalities take place
\[ \Gamma(-\alpha_R) > 0, \quad  \left[ \nu_R^{\alpha_R} - (\nu_R-1)^{\alpha_R}\right] > 0, \]
\noindent as well as the existence condition in \eqref{whole} requires $\lambda_i(M_1) > 0, \ i=1,N$ with $\lambda_i(M_1)$ being the eigenvalues of matrix $M_1 = \nu_R I - A_1$.

To remind, based on Proposition~\ref{prop0} we want $L_R$ to be the negative of an M-matrix. However, this could not be achieved. Indeed, suppose we chose $A_1 = A^B$. Then matrix $M_1$ is the Metzler matrix, unless $h$ is restricted from the bottom, $h > 1/\nu_R$, which is not a good choice because the accuracy of such a method is also restricted by these values of $h$. But on the other hand at $h < 1/\nu_R$ we break the existence condition because $\lambda_i(M_1) < 0, \ i=1,N$. Thus, $A_1 = A^B$ is not a choice.

Now let us try $A_1 = A^F$. Then $M_1$ is a bi-diagonal M-matrix with negative elements on the first upper diagonal. Therefore, $M_1^{\alpha_R}$ is an upper triangular matrix, also with negative elements on the first upper diagonal (property 1). Trying to construct $L_R$ to be the negative of an M-matrix we must choose $A_2 = A^F$. But as
\[ \lambda_R \Gamma(-\alpha_R) \left\{ \left(\nu_R +1/h \right)^{\alpha_R} - \nu_R^{\alpha_R} + \left[ \nu_R^{\alpha_R} - (\nu_R-1)^{\alpha_R}\right]/h \right\} > 0, \ \forall h \]
\noindent it is not possible to have the diagonal elements to be non-positive (property 2). Both properties 1 and 2 make it impossible to construct a stable approximation of $L_R$. The effect should be more pronounced when $\alpha_R$ moves from 1 to 2, similar to what was observed in \cite{WangWanForsyth2007}.

The following proposition solves the above problem.
\begin{proposition} \label{alphaR12}
Consider $1 < \alpha_R < 2$. Because the singularity in the CGMY measure has been already integrated out, the last term in the operator $\mathcal{L}_R$ could be taken out of the jump operator and moved to the diffusion part. Suppose that the following discretization scheme for the remaining operator
\[ L_R = \lambda_R \Gamma(-\alpha_R) \left[ \left(\nu_R  - \triangledown\right)^{\alpha_R} - \nu_R^{\alpha_R} \right]
\]
is in order
\begin{align} \label{cgmy2}
 M &= \lambda_R \Gamma(-\alpha_R) \left[ \left(A_2^C + \nu_R^2 I - 2 \nu_R A^C\right)  \left( \nu_RI - A^F_2\right)^{\alpha_R-2} - \nu_R^{\alpha_R} I \right]
 \end{align}
 \noindent where $A^C_2 = A^F\dot A^B$ is the central difference approximation of the second derivative $\triangledown^2$, $A^C = (A^F + A^B)/2$ is the central difference approximation of the first derivative $\triangledown$. Then $M$ is an $O(h^2)$ approximation  of the operator $L_R$ and the negative of an M-matrix.
\end{proposition}
\begin{proof}
See Appendix.
\end{proof}
The trick is that we represent the operator $(\nu_R - \triangledown)^{\alpha_R}$ as
$L_{1R} = (\nu_R - \triangledown)^2 (\nu_R - \triangledown)^{-\varepsilon}$ where $\varepsilon \equiv 2-\alpha_R$. The first multiplier in $L_{1R}$ is a convection-diffusion operator, and we use a well-known central difference approximation of the second order to discretize this part. The second multiplier is similar to
$(\nu_R - \triangledown)^{\alpha_R}$ in the case $-1 < \alpha < 0$ (because by definition $-1 < -\varepsilon < 0$, and, therefore, we use same discretization as in that case.

To check the convergence numerically we run the same test as for $0 < \alpha_R < 1$, but now choosing  $T = 0.01, \alpha_R = 1.98$. The results are given in Tab.~\ref{tab4}. The "exact" price at $N=2000$ is $C=$ 8.1973. All prices are computed via the matrix exponential.

\begin{table}[h!]
\begin{center}
\begin{tabular}{|c|c|c|c|}
\hline
$C$ & $h$ & $N$ & $\beta$ \cr
\hline
8.2197 & 0.2763100 & 51 &   - \cr
\hline
7.9533 & 0.1381550 & 101 &  3.443 \cr
\hline
8.1558 & 0.0690776 & 201 &  2.557 \cr
\hline
8.1836 & 0.0345388 & 401 &  1.592 \cr
\hline
8.1943 & 0.0172694 & 801 &  2.210 \cr
\hline
8.1970 & 0.0086347 & 1601 & 3.214  \cr
\hline
\end{tabular}
\caption{Convergence of the proposed scheme for the CGMY model with $\alpha_R = 1.98$.}
\label{tab4}
\end{center}
\end{table}
While the convergency ratio $\beta$ looks a bit sporadic, the rate of convergence is closer to 2.

Further analysis of the matrix $M$ reveals two important observations. First, the minimum eigenvalue of $M$ could be close to zero. Therefore, the proposed scheme is close to a family of the $A$-stable schemes, rather than to the $L$-stable ones \footnote{An example of an $A$-stable scheme is the familiar Crank-Nicholson scheme. But we want to underline that 0 doesn't belong to the spectrum of $M$, so formally the scheme is L-stable, while with convergence properties close to the A-stable scheme. The formal L-stability is important, e.g., for computing the option Greeks.}.
Second, the maximum eigenvalue of $e^{\Delta \tau M}$ as $h \rightarrow 0$ tends to 1 which makes the convergence slow, and the conditional number of the matrix high. Also under this situation round-off errors could play a significant role. Performance-wise as it was mentioned in \cite{WangWanForsyth2007} Picard iterations in this case converge very slow and, therefore, direct computation of the matrix exponential (this step could be pre-computed) followed by computation of the product of matrix by vector could be preferable. Our experiments show that the necessary number of iterations could exceed 30. A simple calculus shows that two FFT with the total number of nodes $N$=3000 (including the extended grid to avoid wrap-around effects) gives complexity $O(2 \cdot 30 \cdot N \log_2N ) \propto 2\cdot 10^6$ which corresponds to the complexity of multiplication of a $N \mbox{x}N$ matrix by a $N\mbox{x}1$ vector with $N=1400$. Also if a uniform grid is used, matrix $e^{\Delta \tau M}$ is the Toeplitz matrix, therefore the FFT algorithm for computing a matrix by vector product is applied. Also as shown in \cite{WangWanForsyth2007} values obtained at a non-uniform grid could be re-interpolated (with complexity $O(N)$) to the uniform grid, so again FFT can be applied for the matrix-vector multiplication followed by the back interpolation to the non-uniform grid.

\subsubsection{Approximations of $\mathcal{L}_L$}
Approximations to $\mathcal{L}_L$ can be constructed in a way similar to those corresponding to $\mathcal{L}_R$. Below we will present a few propositions that specify our construction. Proofs of these propositions are omitted because they are very similar to that for $\mathcal{L}_R$.

\begin{proposition} \label{alphaL}
If $\alpha_L < 0$, then the discrete counterpart $L_L$ of the operator $\mathcal{L}^*_L$ which is
$\mathcal{L}_L$ with the "drift" term moved to the diffusion part, is the negative of an M-matrix if
\[ L_L = \lambda_L \Gamma(-\alpha_L) \left\{ \left(\nu_L I + A^B_2\right)^{\alpha_L} - \nu_L^{\alpha_L} \right\}.\]
The matrix $L_L$ is an $O(h^2)$ approximation  of the operator $\mathcal{L}^*_L$.
\end{proposition}

\begin{proposition} \label{alpha01L}
If $0 < \alpha_L < 1$, then the discrete counterpart $L_L$ of the operator $\mathcal{L}^*_L$ which is
$\mathcal{L}_L$ with the "drift" term moved to the diffusion part, is the negative of an M-matrix if
\[ L_L = \lambda_L \Gamma(-\alpha_L) \left\{ \left(\nu_L I + A^B\right)^{\alpha_L} - \nu_L^{\alpha_L} \right\}.\]
The matrix $L_L$ is an $O(h)$ approximation  of the operator $\mathcal{L}^*_L$.
\end{proposition}

\begin{proposition} \label{alpha1L}
Suppose $\alpha_L = 1$ and consider the following discrete approximation of $\mathcal{L}_L$:
\[
L = \lambda_L \Big\{ (\nu_L+A^B)\log (\nu_L+A^B) - \nu_L \log (\nu_L)I - \kappa A^B
\left[(\nu_L+1)\log (\nu_L+1) - \nu_L \log \nu_L \right]  \Big\}.
\]
\noindent where $\kappa$ is some constant. This approximates the operator $\mathcal{L}_L$ with $O(h)$, and is the negative of an M-matrix.
\end{proposition}

\begin{proposition} \label{alphaL12}
Consider $1 < \alpha_L < 2$. Because the singularity in the CGMY measure has been already integrated out, the last term in the operator $\mathcal{L}_L$ could be taken out of the jump operator and moved to the diffusion part. Suppose that the following discretization scheme for the remaining operator
\[ L_L = \lambda_L \Gamma(-\alpha_L) \left[ \left(\nu_L + \triangledown\right)^{\alpha_L} - \nu_L^{\alpha_L} \right]
\]
is in order
\begin{align} \label{cgmy2l}
 M &= \lambda_L \Gamma(-\alpha_L) \left[ \left(A_2^C + \nu_L^2 I + 2 \nu_L A^C\right)  \left( \nu_L I + A^B_2\right)^{\alpha_L-2} - \nu_L^{\alpha_L} I \right]
 \end{align}
 \noindent where $A^C_2 = A^F\dot A^B$ is the central difference approximation of the second derivative $\triangledown^2$, $A^C = (A^F + A^B)/2$ is the central difference approximation of the first derivative $\triangledown$. Then $M$ is an $O(h^2)$ approximation  of the operator $L_L$ and the negative of an M-matrix.
\end{proposition}

\section{Other numerical experiments} \label{Sec6}
In this section we provide a numerical solution of the whole problem (not just one step) to compare it with the existing analytical one. In the first test we used our numerical approach to compute prices of European vanilla options under the Bates model (a Heston jump-diffusion model with Merton's jumps). This solution was compared with the semi-analytical solution obtained by using an inverse Fourier Transform (FFT) since the characteristic function for the Bates model is known in closed form; see, e.g., \cite{Crepey2000}.

For the diffusion step we used the method described in detail in \cite{HoutFoulon2010}. A nonuniform space grid was constructed in both $x$ and $v$ dimensions which contained 100 nodes in $x \in [0,S_{max}], \ S_{max} = 40 \max(S_0,K)$, and 40 nodes in $v \in [0, v_{max}], \ v_{max} = 5 v_0$. Here $K$ is the strike, $S_0, v_0$ are the initial levels of the stock price and instantaneous variance. For the jump step this grid was extended to $S_{up} = 10^4$. Further increase of $S_{up}$ does not influence the option price much, so this boundary was chosen based on a practical argument. The steps of the jump grid when outside of the diffusion grid (where they both coincide with each other) grew according to geometric progression $h_i = h \times g^i$, where $h = (S_{max} - S_{min})/N$ is an average step size for the diffusion grid, $g$ is the growth factor, which in our experiments was chosen as $g=1.03$. The total jump grid thus contained 237 nodes, 75 of which were the diffusion grid nodes.

The initial parameters used in the test are given in Table~\ref{TabParam}. Here $C$ stays for a call option while $P$ for a put option, $r$ is the interest rate, $q$ is the dividend yield, $\kappa$ is the mean-reversion rate, $\xi$ is the volatility of volatility, $\rho$ is the correlation coefficient, $\theta$ is the mean-reversion level.
\begin{table}[H]
\begin{center}
\begin{tabular}{|c|c|c|c|c|c|c|c|c|c|c|c|c|}
\hline
Test & $T$ & $K$ & $r$ & $q$ & $C/P$ & $\xi$ & $\rho$ & $\kappa$ & $\theta$ & $\lambda$ & $\mu_J$ & $\sigma_J$   \cr
\hline
1 & 1 & 100 & 0.05 & 0.0 & C & 0.3 & -0.5 & 1.5 & 0.1 & 5 & 0.3 & 0.1   \cr
 \hline
\end{tabular}
\end{center}
\caption{Initial parameters used in test calculations.}
\label{TabParam}
\end{table}

We computed European option prices under the Bates model in two ways. The first approach utilizes the fact that the characteristic function of the Bates model is known in closed form. Therefore, pricing of European options can be done using any FFT algorithm. Here we used a standard version of the \cite{CarrMadan} method with a constant dumping factor $\alpha = 1.25$ and $N=8192$ nodes. The second approach (FDE) uses an algorithm described in this paper, i.e., splitting and matrix exponentials, where the diffusion (Heston) equation was solved using the method of fractional steps described in \cite{HoutFoulon2010}.

In Fig.~\ref{Fig1} absolute and relative differences in prices obtained in our experiments are presented as a function of moneyness $M = S_0/K$. It is seen that the relative differences between the FDE prices and that obtained with the FFT method are about 0.2\% for ITM options with $1 < M < 1.4$, while they drop down  to 0.8\% for $M = 0.5$\footnote{As it was mentioned in Introduction, in this particular case FFT is definitely more efficient, so we provide this comparison just for illustrative purposes.}

\begin{figure}[ht]
\begin{minipage}{0.48\linewidth}
\begin{center}
\fbox{\includegraphics[width=3.1 in]{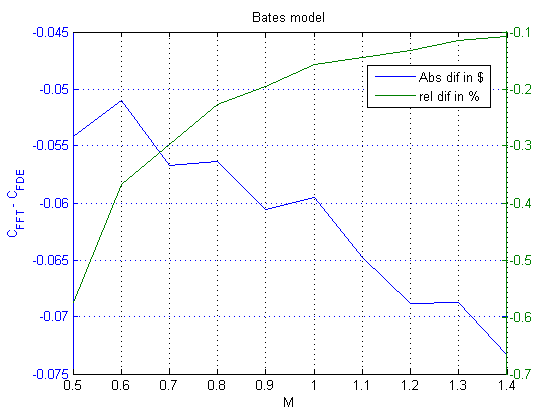}}
\caption{Absolute and relative differences in call option price as a function of moneyness $M$ for the Bates model computed using an FFT algorithm (FFT) and the algorithm of this paper (FDE).}
\label{Fig1}
\end{center}
\end{minipage}
\hspace{0.04\linewidth}
\begin{minipage}{0.48\linewidth}
\vspace{-0.2 in}
\begin{center}
\fbox{\includegraphics[width=3.1 in]{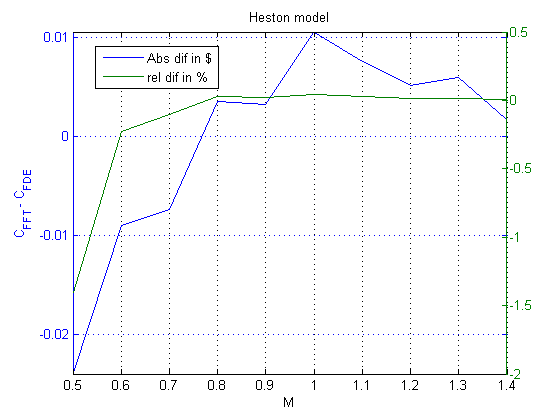}}
\caption{Absolute and relative differences in call option price as a function of moneyness $M$ for the Heston model computed using an FFT algorithm and FDE.}
\label{Fig3}
\end{center}
\end{minipage}
\end{figure}

To see how much of the observed numerical error could be attributed to the Heston model itself, e.g., to the finite-difference algorithm for computing a pure diffusion part, we repeated this test with no jumps and presented these results
in Fig.~\ref{Fig3}.

In the second test we considered a model similar to Bates, but with jumps simulated using the VG model. We used the parameters in Table~\ref{TabParam}. In addition, the VG model parameters were chosen as: $\theta =    0.1, \sigma = 0.4, \nu = 3$, which translates\footnote{For explicit formulae to provide this translation, see \cite{Madan:1989}.} to $\nu_R = 1.5098, \nu_L = 2.7598, \lambda_R = \lambda_L = 0.33$. The grid was constructed as it was in the previous test. However the upper boundary of the jump grid was moved to 10$^5$, and $S_{max} = 20 \max(S_0,K)$. Again we computed European option prices in two ways. As the characteristic function of the VG model is known in closed form, the characteristic function of our model is a product of that for the Heston and VG models. We then used an FFT algorithm proposed by Alan Lewis, and as applied to the VG model discussed in detail in \cite{ItkinVG}. The second approach uses the algorithm described in this paper.

In Fig.~\ref{Fig2}, the absolute and relative differences in prices obtained by these two methods are presented as a function of the moneyness $M = S_0/K$. Here FDE behaves worse than in the case of the Bates model, because we used just the first order approximation in $h$. Still, the relative difference with the FFT solution is less than 0.5\%, and for $M \approx 0.5$ the difference rises to only 1.7\%.

\begin{figure}[h]
\begin{center}
\fbox{\includegraphics[width=3.5in]{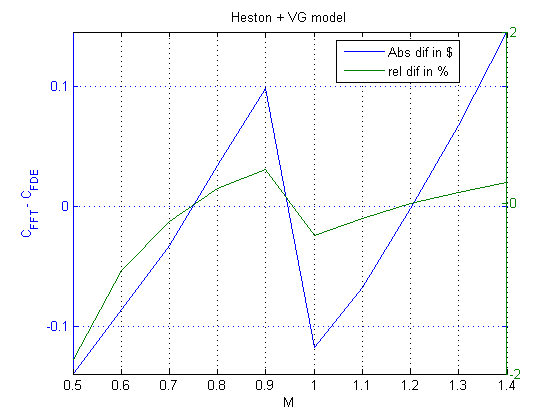}}
\caption{Absolute and relative differences in call option price as a function of moneyness $M$ under the Heston+VG model computed using Lewis's FFT algorithm and FDE.}
\label{Fig2}
\end{center}
\end{figure}

\section{Conclusion}

In this paper (which is a further extension of our paper \cite{ItkinCarr2012Kinky}) we proposed a new method to solve jump-diffusion PIDEs. This method exploits a number of ideas, namely:

\begin{enumerate}
\item First, we transform a linear non-local integro-differential operator (jump operator) into a local nonlinear (fractional) differential operator. Thus, the whole jump-diffusion operator $\mathcal{J} + \mathcal{D}$ is represented as a sum of the linear and non-linear parts.
\item Second, operator splitting on financial processes\footnote{This is similar to splitting on physical processes, e.g., convection and diffusion, which is well-known in computational physics.} is applied to this operator, namely splitting a space operator into diffusion and jumps parts. For nonlinear operators, this approach was elaborated on based on the definition of Lie derivative (see \cite{ThalhammerKoch2010}). The described splitting scheme provides a second-order approximation of $\mathcal{J} + \mathcal{D}$ in time.
\item At the third step various finite-difference approximations of the non-linear differential operator $\mathcal{J}$ are proposed for the Merton, Kou and GTSP (a.k.a., CGMY or KoBoL) models. We demonstrated how to construct these approximations to (i) be unconditionally stable, (ii) be of first- and second-order approximation in the space grid step size $h$ and (iii) preserve positivity of the solution. The results are presented as propositions, and the corresponding proofs are given based on modern matrix analysis, including a theory of M-matrices, Metzler matrices and eventually exponentially nonnegative matrices.
\item It is shown that with a minor modification the method could be applied to the CGMY model with parameter $\alpha > 1$. That is the range where similar algorithms, e.g., \cite{WangWanForsyth2007} experienced a problem. We show how to construct the second order approximation and provide the results of numerical experiments that confirm the second order convergence. Performance-wise matrix exponential followed by computation of a product of matrix by vector seems to be a preferable choice in this case as Picard iterations converge very slow. That is because the maximum eigenvalue of the transition matrix in this case is close to 1. Also under this condition the round-off errors could be important.
\end{enumerate}

All these results seem to be new. The method is naturally applicable to both uniform and nonuniform grids, and is easy for programming, since the algorithm is similar to all jump models. Also notice that the present approach allows pricing some exotic, e.g., barrier options as well. In addition, it respects not just vanilla but also digital payoffs. In principle, American and Bermudan options could also be priced by this method, however this requires some more delicate consideration which will be presented elsewhere.

\section*{Acknowledgments}
I thank Peter Carr and Peter Forsyth for very fruitful discussions, and Igor Halperin and Alex Lipton for useful comments. Also various suggestions of three anonymous referees significantly improved the paper, so their work is appreciated. I am indebted to Gregory Whitten, Steven O'Hanlon and Serguei Issakov for supporting this work, and to Nic Trainor for editing the manuscript. I assume full responsibility for any remaining errors.

\clearpage
\bibliographystyle{apalike}

\newcommand{\noopsort}[1]{} \newcommand{\printfirst}[2]{#1}
  \newcommand{\singleletter}[1]{#1} \newcommand{\switchargs}[2]{#2#1}

\newpage
\appendix
\appendix
\appendixpage
\section{Proof of Proposition~\protect{\ref{alphaR12}}}
To prove this proposition we need technique which is closely related to the concept of an ``eventually positive matrix''; see \cite{Noutsos2008}. Below we reproduce some definitions from this paper necessary for our further analysis.

\begin{definition} \label{def1}
An $N\times N$ matrix $A = [a_{ij}]$ is called
\begin{itemize}
\item {\it eventually nonnegative}, denoted by $A \overset{v}{\ge} 0$, if there exists a
positive integer $k_0$ such that $A^k \ge 0$ for all $k > k_0$; we denote the
smallest such positive integer by $k_0 = k_0(A)$ and refer to $k_0(A)$ as the power index
of $A$;

\item {\it exponentially nonnegative} if for all $t > 0, \ e^{t A} = \sum_{k=0}^\infty \frac{t^k A^k}{k!} \ge 0$;

\item {\it eventually exponentially nonnegative} if there exists $t_0 \in [0,\infty)$ such
that $e^{t A} \ge 0$ for all $t > t_0$. We denote the smallest such nonnegative
number by $t_0 = t_0(A)$ and refer to it $t_0(A)$ s the exponential index of $A$.
\end{itemize}
\end{definition}

We also need the following Lemma from \cite{Noutsos2008}:
\begin{lemma} \label{lemma1}
Let $A \in \mathbb{R}^{N\times N}$. The following are equivalent:
\begin{enumerate}
\item $A$ is eventually exponentially nonnegative.
\item $A + b I$ is eventually nonnegative for some $b \ge 0$.
\item $A^T + b I$ is eventually nonnegative for some $b \ge 0$.
\end{enumerate}
\end{lemma}

We also introduce a definition of an EM-matrix, see \cite{ElhashashSzyld2008}.
\begin{definition} \label{def1}
An $N\times N$ matrix $A = [a_{ij}]$ is called an EM-Matrix if it can be represented as $A = sI - B$ with $0 < \rho(B) < s$, $s > 0$ is some constant, $\rho(B)$ is the spectral radius of $B$, and $B$ is an eventually nonnegative matrix.
\end{definition}

For the following we need two Lemmas.
\begin{lemma} \label{lemma2}
Let $A \in \mathbb{R}^{N\times N}$, and $A = \nu_R I - A_2^F$. Then $A$ is an EM-matrix.
\end{lemma}
\begin{proof}
Denote $d_i$ the $i$-th upper diagonal of $A$. So $d_0$ means the main diagonal, etc.

1. First, show that $A^F_2$ is an eventually exponentially nonnegative matrix. To see this use representation $e^{t A_2^F} = [e^{t B}]^{1/(2h}$ where $B$ is an upper tridiagonal matrix with all $d_0$ elements equal to -3, all $d_1$ elements equal to 4, and all $d_2$ elements equal to -1.  Positivity of $e^{t B}$ can be verified explicitly at $t > N$. The intuition behind that is that the elements on $d_2$ are small in absolute values as compared with that of $d_1$. Taking the square of $B$ propagates large positive values on $d_1$ to the diagonal $d_2$. Taking the square of $B^2$ propagates them to $d_3$, etc.

From $h > 0$ it follows that $e^{t A_2^F} \ge 0$, i.e. $A_2^F$ is eventually exponentially nonnegative.

According to Lemma~\ref{lemma1}, the eventual exponential nonnegativity of $A^F_2$ means that there exists $b \ge 0$ such that $A_2^F + b I = \frac{1}{2h}(B + 2 h b I)$ is eventually nonnegative for some $b \ge 0$. Let us denote $B_1 = B + 2 h b I$ and chose $b = 3/(2h) + \epsilon$, where $\epsilon \ll 1$. In practical examples we can choose $\epsilon = 1.e-6$. Then $d_0(B_1) = \epsilon, d_1(B_1) = 2, d_2(B_1) = -1$. It is easy to check that $B_1^(N+3) \ge 0$. Again that is because $d_1(B_1) > 0, |d_1(B_1)| > |d_2(B_1)|$, so taking the square of $B_1$ propagates large positive values on $d_1$ to the diagonal $d_2$, etc. Thus, $A^F_2 + b I$ with $b = 3/(2h) + \epsilon$ is the eventually nonnegative matrix.

2. Represent $A$ as $A = (\nu_R + b)I - (A_2^F + b I)$. Observe, that $\rho(A_2^F + b I) = \epsilon$ and $s = (\nu_R + b) > \epsilon$. Thus, by definition, $A$ is an EM-matrix. $\blacksquare$
\end{proof}

\begin{lemma} \label{lemma3}
The inverse of the matrix $A = (\nu_R + b)I - (A_2^F + b I) \equiv sI - P$ is a nonnegative matrix.
\end{lemma}
\begin{proof}
Observe that all eigenvalues of $P$ are $\lambda_i = \epsilon, \ \forall i \in [1,N]$. Therefore $\rho(P) = \epsilon$. Following \cite{LeMcDonald2006} denote $index_\lambda(A)$ to be the degree of $\lambda$ as a root of the minimal polynomial of $A$. As matrix $P$ doesn't have zero eigenvalues in its spectrum $index_0 (P) = 0 < 1$.

Nonnegativity of $A^{-1}$ then follows from the Theorem
\begin{theorem}[Theorem 4.2 in \cite{LeMcDonald2006}]
Let $P$ be an $N \times N$ irreducible eventually nonnegative matrix with $index_0(P) \le 1$,
then there exists $\mu > \rho(P)$ such that if $\mu > s > \rho(P)$, then $(s I - P)^{-1} \ge 0$.
\end{theorem}
To apply this Theorem choose any positive $\mu > s$.
\end{proof}

Now we are ready to prove the Proposition~\ref{alphaR12}.
\begin{proof}[Proof of Proposition~\ref{alphaR12}]
Recall, that in the Proposition~\ref{alphaR12} the following scheme is proposed in \eqref{cgmy2}
\begin{align*}
 M &= \lambda_R \Gamma(-\alpha_R) \left[ \left(A_2^C + \nu_R^2 I - 2 \nu_R A^C\right)  \left( \nu_RI - A^F_2\right)^{\alpha_R-2} - \nu_R^{\alpha_R} I \right]
\end{align*}

We prove separately each statement of the proposition, namely:
\begin{enumerate}
\item The above scheme is $O(h^2)$ approximation of the operator $ L_R$;
\item Matrix $M$ is the negative of an M-matrix.
\end{enumerate}

{\it Proof of (1):} This follows from the fact that $A^C$ is a central difference approximation of the operator $\triangledown$ to second order in $h$, while $A^F_2$ is the one-sided second order approximation.\\

{\it Proof of (2):} Matrix $M_1 = A_2^C + \nu_R^2 I - 2 \nu_R A^C$ has the following elements: $\nu_R^2 - \frac{2}{h^2}$ on the main diagonal, $\frac{\nu_R}{h} + \frac{1}{h^2}$ on the first lower diagonal,
and $ - \frac{\nu_R}{h} + \frac{1}{h^2}$ on the first upper diagonal. At small enough $h$ this is the negate of an M-matrix.

Matrix $M_2 = \nu_R - A^F_2$ by Lemma~\ref{lemma2} is an EM-matrix. Now observe that:
\begin{enumerate}
\item As $1 < \alpha_R < 2$, so $-1 < k < 0$.
\item The inverse of an EM-matrix $M_2$ is a nonnegative matrix, see Lemma~\ref{lemma3}.
\item A $k$ power of a nonnegative matrix with $ 0 < k < 1$ is a nonnegative matrix.
\end{enumerate}
A product of the nonnegative and the Metzler matrix is the Metzler matrix. Therefore, $M_1 M_2$ is the Metzler matrix, and so is $M = M_1 M_2 - \nu_R^{\alpha_R} I$. Since coefficient $\lambda_R \Gamma(-\alpha_R) > 0$ at $1 < \alpha_R < 2$, the entire matrix $L_R$ is the negative of an M-matrix.
That finalizes the proof. $\blacksquare$
\end{proof}

\end{document}